\documentclass[10pt,letter]{IEEEtran}
\pdfoutput=1
\usepackage[usenames,dvipsnames]{xcolor}
\usepackage{epsfig}
\usepackage{graphics}
\usepackage{cite}
\usepackage{amsmath}
\usepackage{amssymb}
\usepackage{amsthm}
\usepackage{amsfonts}
\usepackage{tikz}
\usepackage{setspace}
\usepackage{color}
\usepackage{caption}
\newtheorem{theorem}{Theorem}
\newtheorem{proposition}{Proposition}
\newtheorem {result} {Result}
\newtheorem{lemma}{Lemma}
\newtheorem{remark}{Remark}
\newtheorem{corollary}{Corollary}[theorem]

\usepackage{relsize}
\usepackage{bigints}
\usepackage{mathrsfs}
\usetikzlibrary{shapes,arrows,fit,calc,positioning}
\usetikzlibrary{decorations.shapes}
\definecolor{colorgreen}{rgb}{0.2,0.5,0.3}
\definecolor{colorblue}{rgb}{0.2,0.1,0.7}
\usetikzlibrary{decorations.pathmorphing}
\newcommand{\Define}{\triangleq}

\usepackage{subfigure}
\usepackage{epstopdf}
\usepackage{ellipsis}
\usepackage{bm}
    
\IEEEoverridecommandlockouts
\begin{document}
\title{Achievable Rate Region of the Zero-Forcing Precoder in a  2 $\times$ 2 MU-MISO Broadcast VLC Channel with Per-LED Peak Power Constraint and Dimming Control}
%\author{\IEEEauthorblockN{Anup K. Mandpura, \emph{Student Member,~IEEE},
% and Shankar Prakriya, \emph{Senior Member,~IEEE}}}
\author{\IEEEauthorblockN{Amit Agarwal
 and Saif Khan Mohammed\thanks{The authors are with the Department of Electrical Engineering, Indian Institute of Technology Delhi (IITD), New delhi, India. Saif
 Khan Mohammed is also associated with Bharti School of Telecommunication Technology and Management (BSTTM), IIT Delhi. Email: saifkmohammed@gmail.com. This work is supported by the Visvesvaraya Young Faculty Research Fellowship (YFRF) of the Ministry of Electronics and Information Technology, Govt. of India.}}}
%\author{Amit Agarwal,~Sudarshan Mukherjee
% and ~Dr. Saif K. Mohammed \vspace{-1 cm}}
\maketitle 
\begin{abstract}
%In this paper, we consider the $2 \times 2$ multiple-user multiple-input-single-output (MU-MISO) broadcast visible light communication (VLC) channel with two light emitting diodes (LEDs) at the transmitter and single photo diode (PD) at each of the two users. We propose an achievable rate region for this $2\times 2$ MU-MISO VLC channel under a per LED peak and average power constraint, where the average optical power emitted from each LED is fixed for constant lighting, but is controllable (referred to as dimming control in IEEE 802.15.7 standard on VLC). We analytically characterize the proposed rate region boundary and show that it is pareto-optimal. We also propose a novel transceiver architecture where the channel encoder and dimming control are separated which greatly simplifies the complexity of the transceiver. A case study of an indoor VLC channel reveals that the achievable information rates are sensitive to the placement of the LEDs and the PDs. An interesting observation is that for a given placement of LEDs in a 5 m $\times$ 5 m $\times$ 3 m room, even a substantial displacement of the users by 80 cm from their optimal placement results in less than 20 percent reduction in the achievable symmetric rate for a dimming requirement of 30 percent. This allows for substantial mobility of the users around their optimal placement. Results derived in this paper could therefore be used to define ``coverage zones'' within a room where the reduction in the information rates to the two users is within an acceptable tolerance limit.
% % % % % % % % % % % % % % % % % %\
In this paper, we consider the 2 $\times$ 2 multi-user multiple-input-single-output (MU-MISO) broadcast visible light communication (VLC) channel with two light emitting diodes (LEDs) at the transmitter and a single photo diode (PD) at each of the two users. We propose an achievable rate region of the Zero-Forcing (ZF) precoder in this 2 $\times$ 2 MU-MISO VLC channel under a per-LED peak and average power constraint, where the average optical power emitted from each LED is fixed for constant lighting, but is controllable (referred to as dimming control in IEEE 802.15.7 standard on VLC). We analytically characterize the proposed rate region boundary and show that it is Pareto-optimal. Further analysis reveals that the largest rate region is achieved when the fixed per-LED average optical power is half of the allowed per-LED peak optical power.  We also propose a novel transceiver architecture where the channel encoder and dimming control are separated which greatly simplifies the complexity of the transceiver. A case study of an indoor VLC channel with the proposed transceiver reveals that the achievable information rates are sensitive to the placement of the LEDs and the PDs. An interesting observation is that for a given placement of LEDs in a 5 m $\times$ 5 m $\times$ 3 m room, even with a substantial displacement of the users from their optimum placement, reduction in the achievable rates is not significant. This observation could therefore be used to define ``coverage zones'' within a room where the reduction in the information rates to the two users is within an acceptable tolerance limit.

%For this placement we found that the maximum achievable symmetric rate depends on the channel gains only through their difference. We also found that for a fixed LED separation, maximum possible symmetric rate first increases with increasing user separation and then decreases for sufficiently large user separation. We numerically compute the optimum user distance at which they can communicate at maximum achievable symmetric rate, at a fixed LED separation and for co-planar and symmetric physical setting. 
% GRAB250
\end{abstract}
\begin{IEEEkeywords}
Visible light communication, rate region, zero-forcing, multi-user, multiple-input-multiple-output.\vspace{-0.2 cm}
\end{IEEEkeywords}

\section{Introduction}
%Visible Light Communication (VLC) is an optical wireless communication (OWC) technology, which uses visible band of electro-magnetic spectrum for transmission. This technology utilizes existing lighting infrastructure to provide high speed indoor wireless data transmission. This is a complementary wireless communication technology to conventional radio frequency (RF) technology,  as visible light does not interfere with the radio waves. Similar to RF, in VLC, we use multiple LEDs to transmit data to multiple number of users, thus forming a multi user multi-input-single-output (MU-MISO) broadcast system. 

Visible light communication (VLC) is a form of optical wireless communication (OWC) technology which can provide high speed indoor wireless data transmission using existing infrastructure for lighting. One distinctive advantage of VLC technology is that it utilizes the unused visible band of the electromagnetic spectrum and does not interfere with the existing radio frequency (RF) communication in the UHF (Ultra High Frequency) band \cite{Haas_ieee_comm_magzine},\cite{hass_book}. \par 

In VLC systems, it is common to use intensity modulation (IM) via light emitting diode (LEDs) for transmission of information signal and direct detection (DD) via photodiodes (PDs) for the recovery of the information signal \cite{Haas_ieee_comm_magzine}. Contrary to RF systems, in VLC systems the modulation symbols must be non-negative and real valued as information is communicated by modulating the power/intensity of the light emitted by the optical source (LED). The modulation symbols are also constrained to be less than a pre-determined value as the intensity of the light emitted by the LED is peak constrained due to safety regulations and also due to the limited linear range of the transfer function of LEDs \cite{Haas_ieee_comm_magzine},\cite{barry_infrared}. Moreover, due to constant lighting  the mean value of the modulation symbol is also fixed (i.e., non-time varying) and can be adjusted according to the users' requirement (dimming target) \cite{wang_dimming_analysis_PAM}, \cite{wang_tight_bound_dimmable}. \par 

Due to these constraints, analysis performed for RF systems is not directly applicable to VLC systems. For example, the capacity of the RF single-input-single-output~(SISO) additive white Gaussian noise (AWGN) channel is well known and it has been shown that the \emph{Gaussian} input distribution is capacity achieving. For the case of the optical wireless AWGN SISO channel with IM/DD  transceiver, closed form expression for the capacity is still not known, though several inner and outer bounds have been proposed \cite{lapidothpaper,fano,thangaraj}. However, it has been shown that the capacity achieving input distribution for the IM/DD SISO AWGN optical wireless  channel is discrete \cite{smith}, and has been computed numerically in \cite{icc_2009}. Similarly, for the case of dimmable VLC IM/DD SISO channel with peak constraint, there is no closed from expression for the capacity. However following a similar approach as in  \cite{lapidothpaper}, an upper and lower bound is presented in \cite{wang_capacity_lower_upper_like_lapidoth}. \par
Recently, there has been a lot of interest in multi-user multiple-input multiple-output/single-output (MU-MIMO/MISO) VLC systems, where multiple LEDs are used for information transmission to multiple non-cooperative PDs (users) \cite{fathhass_mimo_comparison},\cite{Haas_ieee_comm_magzine}. Such systems have been shown to enhance the system  sum rate when compared to SISO VLC systems \cite{sumrate_mumimo},\cite{mu-misosumrate_max} .\par 
In \cite{sumrate_mumimo}, the information sum rate of MU-MIMO VLC broadcast systems has been studied under the non-negativity constraint on the signal transmitted from each LED, and also a per-LED average transmitted power constraint with no dimming control. The block diagonalization precoder in \cite{sumrate_mumimo} is used to suppress the multi-user interference and the numerically computed achievable sum rate is shown to be sensitive to the placement of the users and the rotation of the PDs. However, they do not consider peak power constraints which is important due to eye safety regulations and also due to the requirement of limited interference to other VLC systems. \par 
Per-LED peak and average power constraint has been considered in \cite{mu-misosumrate_max}, where the sum-rate of the zero forcing (ZF) precoder is maximized in a IM/DD based MU-MIMO/MISO VLC systems. However, in many practical scenarios fairness is required and therefore maximizing the sum rate might not always be the desired operating regime. For example we would like to find the maximum possible rate such that each user gets the same rate. Such operating points can only be obtained from the rate region characterization of the MU-MIMO VLC systems.
In \cite{Chabban_2users}, authors have proposed inner and outer bounds on the capacity region of a two user IM/DD broadcast VLC system where the transmitter has a single LED and each user has a single PD. Per-LED average and peak power constraints are considered. The authors have extended their work to more than two users in \cite{Chabban_multiuser}. However, in both \cite{Chabban_2users} and \cite{Chabban_multiuser}, the transmitter has \emph{only one LED}. Furthermore, dimming control is not considered in   \cite{sumrate_mumimo,mu-misosumrate_max,Chabban_2users,Chabban_multiuser}. \par
The capacity/achievable rate region of a IM/DD based VLC broadcast channel where the transmitter has $N > 1$ LEDs and $M > 1$ users having one PD each, is still an open and challenging problem, primarily due to the non-negativity, peak and average constraints on the electrical signal input to each LED.\par 
In this paper, we consider the smallest instance of this open problem along with dimming control, i.e., with $N=2$ LEDs at the transmitter and $M=2$ users (each having one PD). Dimming control is required in indoor VLC systems since the illumination should \emph{not} vary with time on its own and should be controllable by the users. Therefore, in this paper, in addition to the peak and non-negativity constraints, we constrain the average optical power radiated by each LED to be fixed, i.e., non-time varying. Subsequently in this paper we refer to this system as the $2\times2$ MU-MISO VLC broadcast system.\par
The major contributions of this paper are as follows:
\begin{enumerate}
\item In Section III, we propose an achievable rate region for the $2 \times 2$ MU-MISO VLC broadcast system with the ZF precoder. In this section through analysis we show that the per-LED non-negativity and peak constraint restricts the information symbol vector for the two users (i.e., $(u_1, u_2))$ to lie within a parallelogram $R_{//}$. Each achievable rate pair $(R_1, R_2)$ then corresponds to a rectangle which lies within $R_{//}$. The rate $R_i, i=1,2$ to the $i^{th}$ user depends on the length of the rectangle along the $u_i$-axis. Due to the same average optical power constraint at each LED, these rectangles should also have their midpoint (i.e., point of intersection of the diagonals of the rectangle) at a fixed point on the diagonal of $R_{//}$ denoted by D.\footnote{Out of the two diagonals of $R_{//}$, we refer to the one which has one end point at the origin $(u_1,u_2)$ = (0,0).} This fixed point D on the diagonal of $R_{//}$ is non-time varying, but can be controlled by the user depending upon the illumination requirement. This feature of the proposed system enables dimming control. 
\item In Section III, We also mathematically define the proposed rate region of the ZF precoder for a fix dimming target.
\item In Section IV, we analytically characterize the boundary of the proposed rate region by deriving explicit expressions for the largest possible length along the $u_2$-axis of some rectangle inside $R_{//}$ whose midpoint coincides with the fixed point D on the diagonal of $R_{//}$ and whose length along the $u_1$-axis is given. \emph{Through analysis we also show that the rate region boundary is Pareto-optimal}. 
\item We also analyze the variation in the rate region with change in the dimming level. In depth analysis reveals that the largest rate region is achieved when the fixed point D lies at the midpoint of the diagonal of $R_{//}$, i.e., when the fixed per-LED average optical transmit power is half of the per-LED peak optical power. 
\item For practical scenarios with fairness constraints, through analysis we show that the largest achievable rate pair $(R_1,R_2)$ such that $R_2 = \alpha R_1$ is given by the unique intersection of the proposed rate region boundary with the straight line $R_2 = \alpha R_1$. 
\item In Section V, from the point of view of practical implementation we also propose a novel transceiver architecture where the same channel encoder can be used irrespective of the level of dimming control.
\item Analytical results have been supported with numerical simulations in Section~\ref{sec:numres}. It is observed that for a fixed placement of the two LEDs, the achievable information rates are a function of the placement of the two PDs/users. Specifically, we observe that for a given placement of the two LEDs, there exists an optimal placement of the two users which maximizes the symmetric rate. Another interesting observation is that in a 5 m $\times$ 5 m $\times$ 3 m (height) room with the two LEDs attached to the ceiling and the two PDs placed in the horizontal plane at a height of 50 cm above the floor, even a user displacement of 60 cm from the optimal placement results in only approx. a 10 percent reduction in the symmetric rate when compared to the symmetric rate with the optimal placement of PDs\footnote{For this study the dimming control is such that the average optical power radiated from each LED is 30 percent of the peak allowed optical power}. This allows for substantial mobility of the user terminals around their optimal placement which is specially desirable when the user terminals are mobile/portable. A practical application of the results derived in this paper could be in defining coverage zones for the PDs/users, i.e.,
the maximum allowable displacement of the users for a fixed desired upper
limit on the percentage loss in the achievable information rates.
\end{enumerate}
\vspace{-.2in}
  \begin{figure}[h]
        \centering
      \includegraphics[width=6cm,height=5cm]{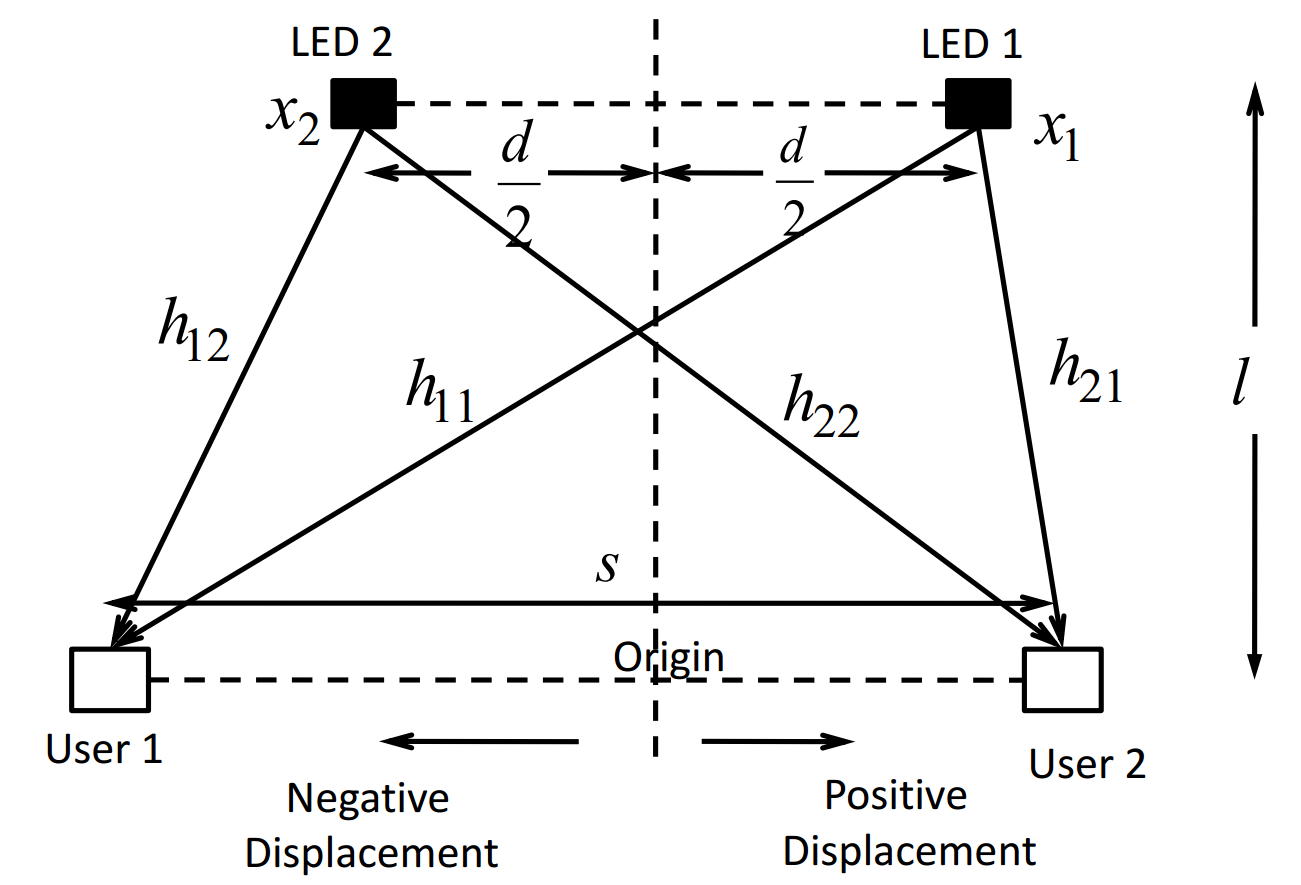}   
      \caption{ $2 \times 2$ IM/DD MU-MISO VLC broadcast system.}
      \label{fig:setup}
      \end{figure}  
      
\section{System Model}
We consider a $2\times 2$ IM/DD MU-MISO VLC broadcast system. The transmitter of the MISO system is equipped with two LEDs and each user has a single photo-diode (PD)~(see Fig.~\ref{fig:setup}).\footnote{ Since each of the two users has a single PD we will be interchangeably using user and PD in subsequent discussions.}
The LED converts the information carrying electrical signal to an intensity modulated optical signal and the PD at each user converts the received optical signal to electrical signal. The transmitter performs beamforming of the information symbols towards the two non-cooperative users. Let $u_1 \in \mathcal{U}_1$ and $u_2 \in \mathcal{U}_2$ be the information symbols  intended for the first and second user respectively, where $\mathcal{U}1$ and $\mathcal{U}_2$ are the information symbol alphabets for user 1 and user~2 respectively. Let $x_i$ be the optical power transmitted from the $i^{\text{th}}$ LED ($i = 1, 2$). At any time instance, the transmitted optical power vector $\bm x \triangleq [x_1~x_2]^T$ is given by \begin{equation} \label{eq:tX_op_power} \bm x = \bm A \bm u, \end{equation} where $\bm u \triangleq [u_1~u_2]^T$ and $\bm A \in \mathbb{R}^{2 \times 2}$ is the beamforming matrix. In this paper, we consider the following power constraints for our dimmable VLC system. \par The instantaneous power transmitted from each LED is non negative and is less than some maximum limit $P_0$ due to skin and eye safety regulations \cite{barry_infrared}. Further, such a maximum limit on the transmitted power is required also due to limited interference requirement to the neighboring VLC systems, i.e. \begin{equation}\label{eq:peakpowerconst}0\leq x_i \leq P_0,~i = 1,2. \end{equation} \par Since our VLC system is dimmable we further impose a per-LED average power constraint of the type  \begin{equation} \label{eq:dimming_const} E[x_i] = \xi P_0,~i = 1,2, \end{equation} where $0 \leq \xi \leq 1$ is the dimming target \cite{wang_tight_bound_dimmable}. For the sake of analysis, we define $x_i^{\prime} \triangleq \frac{x_i}{P_0},~i=1,2$ as the normalized power transmitted from each LED. Consequently, the normalized optical power transmitted from each LED must satisfy the following constraints given by 
\begin{align}\label{eq:constraints_normaized}
0\leq x_i^{\prime} \leq 1~~ \& ~~~
%E[x_i^{'}] = \xi,~i = 1,2 
 E[x_i^{\prime}] = \xi,~i = 1,2.
\end{align}
Assuming $y_k$, $k = 1, 2$ to be the normalized received electrical signal at the $k^{\text{th}}$ user (after scaling down by $P_0$), the normalized received signal vector is given by\footnote{In subsequent discussions, by ``received electrical signal'', we refer to the ``normalized received electrical signal''.}
\small\begin{align}
\underbrace{\begin{bmatrix}
y_1\\
y_2
\end{bmatrix}}_{\triangleq \bm y} = \underbrace{\begin{bmatrix}
    h_{11} & h_{12}\\
    h_{21} & h_{22}
    \end{bmatrix}}_{\triangleq \bm H} \underbrace{\begin{bmatrix}
    x_1^\prime\\
    x_2^\prime
    \end{bmatrix}}_{\triangleq \bm x^\prime} + \underbrace{\begin{bmatrix}
           n_1\\
           n_2
              \end{bmatrix}}_{\triangleq \bm n}, \nonumber \end{align}
              \begin{equation}\label{eq:rx_signal}
              \text{s.t.}~~0 \leq  x_i^\prime \leq 1, ~~~~  E[x_i^{\prime}] = \xi, ~~~  i=1,2.
\end{equation}\normalsize
\noindent where $\bm H \triangleq [h_{ki}]_{2 \times 2}$ is the channel gain matrix. The channel gain coefficients between the $i^{th}$ LED and the $k^{th}$ user is denoted by $h_{ki}, i=1,2, k=1,2.$\footnote{\footnotesize{Note that $h_{ki}$'s are non negative and model the overall gains of the line of sight (LOS) optical path  between the $i^{th}$ LED and the $k^{th}$ user and also the responsivity of the PD of the $k^{th}$ user \cite{barry_infrared}.}} We further define $\bm h_1\Define [h_{11}~h_{21}]^T$ and $\bm h_2 \Define [h_{12}~h_{22}]^T$ to be the channel vectors from LED 1 and LED 2 respectively. Further, $n_1$ and $n_2 $ are the sum of the thermal noise and ambient light-induced shot noise at the respective users\footnote{Note that the above noise impairments of the received signal are the main impairments that are commonly assumed in VLC systems \cite{barry_infrared}.} and are independent  of $x_1'$ and $x_2'$ \cite{fathhass_mimo_comparison}. The noise signals are i.i.d. zero mean real AWGN with variance
$\sigma^2/P_0^2$, where $\sigma^2$ is the variance of the noise before the scaling down of the received signal by $P_0$, i.e., $n \sim \mathcal{N}(0, (\sigma/P_0)^2).$ 
 \vspace{-.2in}
   \begin{figure}[t]
          \centering
          \includegraphics[width=6cm,height=5cm]{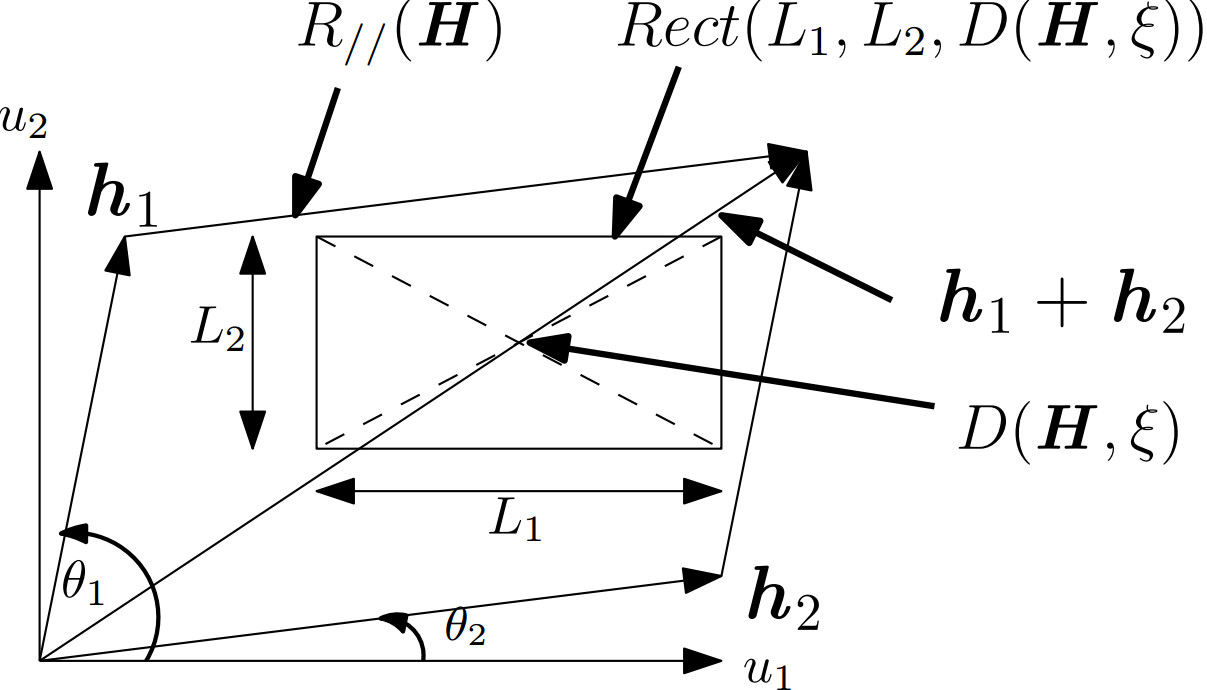}
          \caption{The information vector $\bm u$ is constrained to lie within the parallelogram $R_{//}(\bm H)$ whose non-parallel sides are $\bm h_1$ and $\bm h_2$. The rectangular region $\mathcal{U}_1\times\mathcal{U}_2$ whose length along the $u_1$ axis is  $L_1$ and that along the $u_2$ axis is $L_2$ and whose center lies at $D(\bm H, \xi)$ is denoted by $Rect(L_1,L_2,D(\bm H, \xi))$.}
          \label{fig:rect_def}
           \end{figure}      
 \section{\hspace{-1mm}An achievable rate region of the channel in (\ref{eq:rx_signal})}
 In this section, we derive an achievable rate region for the channel in (\ref{eq:rx_signal}) using the ZF precoder.
 For the $2 \times 2$ MU-MISO system discussed in section II, the ZF precoding matrix is uniquely given by $\bm A = P_0 \bm H^{-1}$, i.e., $\bm x^\prime= {x}/{P_0} = {\bm A u}/{P_0} = {P_0\bm H^{-1}u}/{P_0} = \bm  H^{-1}\bm u$. Thus the received signal vector is given by
 \begin{equation}\label{eq:SISO}
 \bm y = \bm H \bm x^\prime + \bm n = \bm H \bm H^{-1} \bm u + \bm n = \bm u + \bm n.
 \end{equation} 
 i.e., there is no multi-user interference (MUI). Since \begin{equation}
 \label{eq:uhx}
 \bm{u} = \bm H \bm{x^{\prime}} = [\bm h_1~\bm h_2][x_1^{\prime}~x_2^{\prime}]^T, \end{equation} and $0 \leq x_i^\prime \leq 1 , i=1,2$ (see (\ref{eq:constraints_normaized})) it follows that, the information signal vector $\bm u$ must be limited to the region
 \begin{equation}
 \label{eq:parallelogram}
 R_{//}(\bm H)\triangleq  \left\{\bm u~\vert~\bm u = \bm H \bm x',~~0 \leq x_1' \leq 1, 0 \leq x_2' \leq 1 \right\} .\end{equation}   
  The region $R_{//}(\bm H)$ is a parallelogram with its two non parallel sides as $\bm h_1$ and $\bm h_2$ (see $R_{//}(\bm H)$ in Fig.~\ref{fig:rect_def}). In addition to this, the diagonal of the parallelogram $R_{//}(\bm H)$ is the vector $\bm h_1+ \bm h_2$ as shown in Fig.~\ref{fig:rect_def}. 
    \par  Let $E[\bm u] \Define [E[u_1]~ E[u_2]]^T$ be the mean information symbol vector. From (\ref{eq:uhx}) and (\ref{eq:constraints_normaized}), the mean information symbol vector is given by
    \begin{equation}\label{eq:mean_info_vec} E[\bm u] = \bm H E[\bm x'] = [\bm h_1 ~\bm h_2] E[x_1'~x_2']^T \overset{(a)}{=} \xi (\bm h_1+ \bm h_2),\end{equation}
     where step (a) follows from (\ref{eq:rx_signal}). Therefore, the mean information symbol pair $(E[u_1],~ E[u_2])$ is a point corresponding to the tip of the vector $E[\bm u] = \xi (\bm h_1 + \bm h_2)$. From (\ref{eq:parallelogram}) it is clear that the vector $(\bm h_1 + \bm h_2)$ is a diagonal of $R_{//}(\bm H)$ (see Fig. \ref{fig:rect_def}). For a given $0 \leq \xi \leq 1$, the tip of the mean information symbol vector $\xi (\bm h_1 + \bm h_2)$ is therefore a \emph{fixed} point on the diagonal $(\bm h_1 + \bm h_2)$. We denote this point by
     \begin{align} \label{eq:tip}
     D(\bm H,\xi)  &= (E[u_1] , E[u_2]) \nonumber \\
                  &= (\xi (h_{11} + h_{12}) \, , \, \xi (h_{21} + h_{22})).
     \end{align}
   With the ZF precoder, the broadcast channel in (\ref{eq:rx_signal}) is reduced to two parallel SISO (single-input single-output) optical channels between the transmitter and the two users (see (\ref{eq:SISO})). Since $u_1 \in \mathcal{U}_1$ and $u_2 \in \mathcal{U}_2$ are independent and originate from different codebooks, it follows that  $(u_1, u_2) \in  \mathcal{U}_1 \times \mathcal{U}_2$. From (\ref{eq:parallelogram}), we know that $(u_1, u_2)$ must belong to the parallelogram $R_{//}(\bm H)$ and therefore
     \begin{equation}
     \label{eq:EQ2}   
      \mathcal{U}_1 \times \mathcal{U}_2 \subset R_{//}({\bm H}). \end{equation}
      
      In general we choose $\mathcal{U}_1$ and $\mathcal{U}_2$ to be intervals of the type $[a, b]$ \cite{smith}. Let the length of the intervals $\mathcal{U}_1$ and $\mathcal{U}_2$  be $L_1$ and $L_2$ respectively, i.e. $\vert \mathcal{U}_1 \vert = L_1,~ \vert \mathcal{U}_2 \vert = L_2$. With $\mathcal{U}_1$ and $\mathcal{U}_2$ as intervals, it is clear that $\mathcal{U}_1 \times \mathcal{U}_2$ must be a rectangle whose length along the $u_1$ axis is $L_1$ and that along the $u_2$ axis is $L_2$. In this paper we assume $u_1$ and $u_2$, to be \textit{uniformly distributed} in the interval $\mathcal{U}_1$ and $\mathcal{U}_2$ respectively.\footnote{At high SNR ($P_0/\sigma >> 1$), uniformly distributed information symbol is near capacity achieving \cite{icc_2009}.} Therefore, for a given $\mathcal{U}_1$ and $\mathcal{U}_2$, the mean information symbol pair $(E[u_1], E[u_2])$ will lie at the point of intersection of the two diagonal of the rectangle $\mathcal{U}_1 \times \mathcal{U}_2$. We will subsequently call this point of intersection as the ``midpoint'' of the rectangle $\mathcal{U}_1 \times \mathcal{U}_2$ and will denote it by $\mathscr{C}(\mathcal{U}_1,\mathcal{U}_2)$. \par From (\ref{eq:tip}), it follows that the mean information symbol pair must exactly coincide with $D(\bm H,\xi)$, i.e. \begin{equation}\label{eq:cu1u2equal_D} \mathscr{C}(\mathcal{U}_1,\mathcal{U}_2) = D(\bm H,\xi)\end{equation}
            \par
          The ZF precoder transforms the broadcast channel into two parallel SISO channels $y_i = u_i + n_i, u_i \in \mathcal{U}_i, i=1,2.$ Let $R_1$ and $R_2$ denote the information rates achieved on these SISO channels with $u_i$ distributed uniformly in $\mathcal{U}_i$. 
          Any given $\mathcal{U}_1$ and $\mathcal{U}_2$ satisfying the conditions in (\ref{eq:EQ2}) and (\ref{eq:cu1u2equal_D}) would satisfy the optical power constraints in (\ref{eq:constraints_normaized}) and would therefore correspond to an achievable rate pair for the broadcast channel in (\ref{eq:rx_signal}). 
      Since a rectangle in the $u_1-u_2$ plane corresponds to a unique $\mathcal{U}_1$ and $\mathcal{U}_2$ and vice versa, it follows that any rectangle lying inside the parallelogram $R_{//}(\bm H)$ and having its midpoint at $D(\bm H,\xi)$ will correspond to an achievable rate pair. In this paper, for the broadcast channel in (\ref{eq:rx_signal}), we therefore propose an achievable rate region which consists of rate pairs corresponding to such rectangles (one such rectangle is shown in Fig. \ref{fig:rect_def}). We define our proposed rate region more precisely in the following. Towards this end, we first formally define the achievable rate of a SISO AWGN optical channel, where the transmitted information symbol is constrained to lie in an interval.         
       \begin{result}[][From \cite{lapidothpaper}, \cite{icc_2009}]
          The achievable information rate of a SISO channel $y = u + n$ (where $u \sim Unif [a,b]$ and $n \sim \mathcal{N}(0, (\sigma/P_0)^2)$ depends on the interval $[a, b]$ only through its length $L = |b - a|$, and is given by the function
          \begin{equation}
             \label{eq:def_cpapcity}  
          C(L = | b - a|, P_0/\sigma) \Define  I(u; y),  
        \end{equation} 
          here $Unif[a , b]$ denote the uniform distribution in the interval $[a , b]$ and $I(u; y)$ is the mutual information between $u$ and $y$. \end{result}\par
          
          \begin{result}[][From \cite{lapidothpaper}, \cite{icc_2009}] The function $C(L, P_0/\sigma)$ is continuous with respect to $L$ and increases monotonically with increasing $L$ for a fixed $P_0/\sigma$. \end{result}
          
          Let $Rect(L_1, L_2, D(\bm H,\xi)) $ denote the unique rectangle having its midpoint as $D(\bm H,\xi)$ and whose length along the $u_1$ axis is $L_1$ and that along the $u_2$ axis is $L_2$ (see Fig. \ref{fig:rect_def}). Any such rectangle $Rect(L_1, L_2, D(\bm H,\xi)) \subset R_{//}(\bm H)$ will \emph	{correspond} to an achievable rate pair given by
             \begin{align}
             \label{eq:rate_region}
             &(R_{1}, R_{2}) \triangleq  \big(C(L_1, P_0/\sigma), C(L_2, P_0/\sigma)\big)      
             \end{align}      
             For a given $(\bm H, P_0/\sigma,\xi)$ the proposed achievable rate region for the ZF precoder is given by      
               
            \small\begin{align}
                \label{eq:rate_region_equation}
                \hspace{-1mm}R_{\textit{ZF}}(\bm H , P_0/ \sigma, \xi) \triangleq \underset{\overset{}{(L_1,L_2)\in (L_1,L_2) \in S}}{\cup}\hspace{-10mm}\{C(L_1,P_0/\sigma),C(L_2,P_0/\sigma)\},
               \end{align}
                \normalsize
          where {$ S\triangleq\{(L_1 \geq 0,L_2 \geq 0) \vert~ \exists~ Rect(L_1,L_2, D(\bm H, \xi))$ $\subset R_{//}(\bm H)\}$.} 
          \normalsize  
          
 \section{Characterizing the Boundary of the Rate Region $R_{\text{ZF}}(\bm H, P_0/\sigma, \xi)$ }\label{sec:bdzf}
    In this section, we completely characterize the boundary of the rate region, $R_{\textit{ZF}}(\bm H , P_0/ \sigma, \xi)$, for a fixed $(\bm H , P_0/ \sigma, \xi)$. Towards this end, for each information rate $R_{1}$ achievable by the first user, we find the corresponding maximum possible information rate $R_{2}$ achievable by the second user.\textit{ Each pair of $R_{1}$ and its corresponding maximum possible $R_{2}$ is therefore a point on the boundary of the proposed rate region}. By increasing $R_{1}$ from $0$ to its maximum possible value, all such $(R_{1},R_{2})$ pairs characterize the boundary of the rate region.\par
  From (\ref{eq:rate_region_equation}), we know that any achievable rate pair $(R_1, R_2)$ in the proposed rate region $R_{\textit{ZF}}(\bm H, P_0/\sigma, \xi)$ corresponds to some rectangle $Rect(L_1,L_2, D(\bm H,\xi))$. The rate to the $i^{th}$ user, i.e. $R_i=C(L_i, P_0/\sigma), i=1,2$ depends only on the length of this rectangle along the $u_i$-axis. Since the $C(L, P_0/\sigma)$ function is monotonic and continuous in its first argument, each value of $R_i$ corresponds to a unique $L_i$ and vice versa. Therefore, towards characterizing the boundary of $R_{\textit{ZF}}(\bm H, P_0/\sigma, \xi)$, we note that for a given $R_1$,i.e., for a given length $L_1$ along the $u_1$-axis, we would like to find the largest possible $R_2$,i.e., the largest possible $L_2$ such that the rectangle $Rect(L_1, L_2, D(\bm H, \xi))$ lies entirely inside $R_{//}(\bm H)$. Hence, we can characterize the boundary of $R_{\textit{ZF}}(\bm H, P_0/\sigma, \xi)$ simply by varying $L_1=x$ from $0$ to its maximum possible value 
  (denoted by $L_1^{\max}(\xi)$), and for each value of $L_1=x\in[0,L_1^{\max}(\xi)]$ we find the largest possible $L_2 = L^\xi_2(x)$ which gives us a corresponding rate pair $(R_1,R_2) = (C(L_1=x, P_0/\sigma), C(L_2=L^\xi_2(x), P_0/\sigma))$ on the boundary of the rate region $R_{\textit{ZF}}(\bm H, P_0/\sigma, \xi)$.\par  
     For a given ($L_1=x, L_2=L_2^\xi(x)$) the corresponding information rate pair lies on the boundary of the proposed rate region $R_{\textit{ZF}}(\bm H,P_0/\sigma,\xi)$. We denote this information rate pair by $(R_1^{\textit{Bd}}(x, P_0/\sigma,\xi), R_2^{\textit{Bd}}(x, P_0/\sigma,\xi))$. From (\ref{eq:rate_region}), this information rate pair is given by
       \begin{eqnarray}
       \label{eq:rboud1}
       R_1^{\textit{Bd}}(x, P_0/\sigma,\xi)\triangleq C(L=x,P_0/\sigma). 
       \end{eqnarray}
       \begin{eqnarray}
       \label{eq:rboud2}
       R_2^{\textit{Bd}}(x, P_0/\sigma,\xi) \triangleq C(L=L_2^\xi(x),P_0/\sigma).    
       \end{eqnarray}
       \normalsize
       This then completely characterizes the boundary of the rate region  $R_{\textit{ZF}}(\bm H,P_0/\sigma,\xi)$,  which is given by\footnote{
       From (\ref{eq:rboud1}) and (\ref{eq:rboud2}) it is clear that the exact computation of $R_1^{\textit{Bd}}(x, P_0/\sigma,\xi)$ and $R_2^{\textit{Bd}}(x, P_0/\sigma,\xi)$ requires the computation of $L_2^\xi(x)$ for which we derive closed form expressions in the next section. Computation of $R_1^{\textit{Bd}}(x, P_0/\sigma,\xi)$ and $R_2^{\textit{Bd}}(x, P_0/\sigma,\xi)$  also requires us to compute the $C(L,P_0/\sigma)$ function which is done numerically.}   
       \small\begin{align}
       \label{eq:bd_rate_region}
        R_{\textit{ZF}}^{\textit{Bd}}(\bm H,P_0/\sigma,\xi) & \triangleq \hspace{-6mm}\underset {\overset{}{0 \leq x \leq L_1^{\max}(\xi)}}{\cup} \hspace{-3mm}\left(R_1^{\textit{Bd}}(x, P_0/\sigma,\xi),R_2^{\textit{Bd}}(x, P_0/\sigma,\xi)\right)  \nonumber \\
        & = \hspace{-3mm}\underset{\overset{}{0\leq x \leq L_1^{\text{max}}(\xi)}}{\cup} \hspace{-3mm}\big(C(x,P_0/\sigma), C(L_2^\xi(x), P_0/\sigma)\big)
        \end{align}\normalsize
 
  It is noted that the analysis done in this paper is applicable to any placement of the users and the LEDs. 
  Subsequently, we follow the following convention that, by \textit{LED 1 we shall refer to the LED whose channel vector has a higher inclination angle (from the $u_1$ axis) than the inclination angle of the channel vector of the other LED}.\vspace{1mm}\par 
  Let the inclination of the vector $\bm h_1$ and $\bm h_2$ from the $u_1$ axis be $\theta_1$ and $\theta_2$ respectively (see Fig.~\ref{fig:rect_def}). From our definition of LED 1 and LED 2 (see the above paragraph), it follows that $\theta_1 > \theta_2$. Therefore it follows that $\tan\theta_1 > \tan\theta_2$. Since \begin{equation}\label{eq:tantheta1}\tan\theta_1 = h_{21}/h_{11} , ~~~ \tan\theta_2 = h_{22}/h_{12}.\end{equation} Hence, $\tan\theta_1 > \tan\theta_2$ implies that
  \begin{align}\label{eq:dethlesszero}
  h_{21}/h_{11} - h_{22}/h_{12} &> 0 , \nonumber \\
  h_{11}h_{22} - h_{12}h_{21} &<0, ~~\text{i.e.} \nonumber  \\
  det(\bm H) &< 0
  \end{align}    
   In the following proposition, we first compute the maximum value of $L_1$ and subsequently we derive the maximum value of $L_2$ for each value of $L_1$.           
       \begin{proposition} The largest possible value of $L_1$ (i.e., length of the interval $\mathcal{U}_1$) such that  there exists a rectangle $Rect(L_1, L_2, D(\bm H, \xi))~(L_2 \geq 0)$ which lies completely  inside the parallelogram $R_{//}(\bm H)$, is given by \begin{align}
            \label{eq:l1max}
             L_1^{\text{max}}(\xi) &\triangleq \underset{\underset{Rect(L_1,L_2, D(\bm H, \xi))\subset R_{//}(\bm H)}{L_1 \geq 0, L_2 \geq 0 }}{\max}L_1 \nonumber \\ &= \left\{\hspace{-3mm}\begin{array}{ll}  
            \frac {-2 \xi det(\bm H)}{ \max(h_{21}, h_{22})}, & 0 \leq \xi \leq  1/2 \vspace{1mm} \\
                           \frac {-2(1-\xi) det(\bm H)}{ \max(h_{21}, h_{22})},& 1/2 \leq \xi  \leq 1. 
               \end{array} \right.
            \end{align}
         
          \end{proposition}   
          \begin{proof} See Appendix \ref{ap:prop1}. \end{proof}
           It is clear from (\ref{eq:l1max}) in Proposition~1 that $L_1^{\max} (\xi)$ is a continuous function of $\xi$ and $L_1^{\max} (\xi) = L_1^{\max} (1-\xi)$.
          \begin{remark}\label{rem:l1sym}
         The function $L_1^{\max} (\xi)$ is a continuous function of $\xi$ and is symmetric about $\xi=1/2$, i.e.
          \begin{equation}
          \label{eq:symmetric_L1}
          L_1^{\max} (\xi) = L_1^{\max} (1-\xi),~~0\leq\xi \leq 1
          \end{equation}\end{remark}\par  
          From (\ref{eq:l1max}) it is clear that since $det(\bm H) < 0$ (see (\ref{eq:dethlesszero}))  $L_1^{\max}(x)$ is linearly increasing for $0\leq \xi\leq 1/2$ and is linearly decreasing for $1/2 \leq \xi \leq 1$. Hence $L_1^{\max}(\xi)$ has a unique maximum at $\xi=1/2$.
                \begin{remark}\label{rem:l1max} The function $L_1^{\max} (\xi)$ has its unique maximum at $\xi =1/2$, i.e.         
                \begin{equation}
                      \label{eq:max_L1}
                      \underset{0\leq \xi\leq 1}{\arg\max} ~~L_1^{\max} (\xi) = 1/2
                      \end{equation}      
             \end{remark} \par

          \begin{proposition}
             For a given  $L_1 = x\in[0, L_1^{\max}(\xi)]$, the largest possible $L_2 \geq 0$ such that  there exists a rectangle $Rect(x, L_2, D(\bm H, 
             \xi)) \subset R_{//}(\bm H)$, is given by     
             \begin{align} \label{eq:defl2xi}
             L^\xi_2(x) &\triangleq \max_{\underset{Rect(x,L_2,D(\bm H, \xi))\subset \mathcal{R}_{//}(\bm H)}{L_2 \geq 0}} L_2 \nonumber \\
             &=2\min(L^{up,\xi}_2(x), L^{down,\xi}_2(x)), 
             \end{align}
              where $L^{up,\xi}_2(x)$ is given by \vspace{1mm} \\
             \noindent \textit{Case I}: $0 \leq \xi \leq \frac{h_{11}}{h_{11}+h_{12}}$ 
                   \begin{align}
                   \label{eq:l2up1}
                   L^{up,\xi}_2(x)
                         &=\left\{\hspace{-2mm}
                            \begin{array}{ll}
                            \frac{{-\xi det(\bm H)} - {\frac{x}{2} h_{21}}} {h_{11}}, & 0 \leq  x \leq L_1^{\max}(\xi)  
                            \end{array} \right.
                   \end{align}
                  \noindent\textit{Case II}: $\frac{h_{11}}{h_{11}+h_{12}} \leq \xi \leq 1$ 
             \begin{align}\label{eq:l2up2}
             L^{up,\xi}_2(x)
                   &=\left\{\hspace{-2mm}
                      \begin{array}{ll}
                      \frac{{-(1-\xi)det(\bm H)}- {\frac{x}{2} h_{22}}} {h_{12}},  & 0 \leq x \leq \eta_3(\xi)  \vspace{1mm}\\
                      \frac{{-\xi det(\bm H)} - {\frac{x}{2} h_{21}}} {h_{11}}, & \eta_3(\xi) \leq x \leq L_1^{\max}(\xi)                
                      \end{array} \right.
                       \end{align}
             where $\eta_3(\xi) \Define 2\xi h_{12} - 2(1-\xi)h_{11}$. \vspace{2mm} \\ 
                             $L^{down,\xi}_2(x)$ is given by \vspace{2mm}  \\                   
                       \noindent \textit{Case I}: $ 0 \leq \xi \leq \frac{h_{12}}{h_{11}+h_{12}}$ 
           \begin{align}\label{eq:l2down1}
                      L^{down,\xi}_2(x)
                            &=\left\{\hspace{-2mm}
                               \begin{array}{ll}
                               \frac{{-\xi det(\bm H)} - \frac{x}{2}{ h_{22}}} {h_{12}}, &\hspace{-2.5mm} 0 \leq x \leq \eta_4(\xi)  \vspace{1mm} \\ 
                                        \frac{{-(1-\xi) det(\bm H)} - \frac{x}{2} h_{21}} {h_{11}}, & \hspace{-2.5mm}  \eta_4(\xi) \leq x \leq L_1^{\max}(\xi)              
                                         \end{array} \right.
                                        \end{align} 
                                        where $\eta_4(\xi) \Define 2(1-\xi)h_{12}-2\xi h_{11}.$ \vspace{1mm} \\ 
                \noindent \textit{Case II}: $\frac{h_{12}}{h_{11}+h_{12}} \leq \xi \leq  1$              
                                        \begin{align}\label{eq:ledown2}
                                         L^{down,\xi}_2(x)
                                          &=\left\{\hspace{-2mm}
                                           \begin{array}{ll}                                   
                                       \frac { {-(1-\xi) det(\bm H)} - \frac{x}{2}{ h_{21}}} {h_{11}}, & \hspace{-2.5mm}  0\leq x \leq  L_1^{\max}(\xi)                                         
                                         \end{array} \right.
                                         \end{align}
                            
              \end{proposition}
          \begin{proof} See Appendix \ref{ap:prop2}. \end{proof}   
          \begin{lemma}\label{lem:l2_decrease} The function $L_2^\xi(x)~(0 \leq x \leq L_1^{\max}(\xi))$ is a monotonically decreasing and continuous function of $x$. \end{lemma}
          \begin{proof}
           From Proposition 2 it is clear that for a given $\xi$ both $L_2^{up,\xi}(x)$ \text{and} $L_2^{down,\xi}(x)$ are continuous and monotonically decreasing function of  $x$. From this it follows that  $L_2^\xi(x) = 2\min(L^{up,\xi}_2(x), L^{down,\xi}_2(x)) $ is a continuous and decreases monotonically with increasing $x$.   \end{proof}         
    \begin{lemma}\label{lem:pareto} The proposed rate region boundary $R_{\textit{ZF}}^{\textit{Bd}}(\bm H, P_0/\sigma,\xi)$ is Pareto-optimal. That is, for any two rate pairs $(a, b)$ and $(a', b')$ on the boundary $R_{\textit{ZF}}^{\textit{Bd}}(\bm H, P_0/\sigma,\xi)$, if 
                  $a' \geq a$ then it must be true that $b' \leq b$ and if $b' \leq b$ then it must be true that $a' \geq a$. \end{lemma}
                  \begin{proof}
                  Let $(a, b)$ and $(a', b')$ be any two rate pairs on the boundary  $R_{\textit{ZF}}^{\textit{Bd}}(\bm H, P_0/\sigma,\xi)$ such that $a' \geq a$. Then from ((\ref{eq:rboud1}) and (\ref{eq:rboud2})) it follows that there exists $0 \leq x \leq L_1^{\max}(\xi)$ and $0 \leq x' \leq L_1^{\max}(\xi)$ such that $a = C(x,P_0/\sigma), b = C(L_2^\xi(x),P_0/\sigma)$ and $a' = C(x', P_0/\sigma)$, $b' = C(L_2^\xi(x'), P_0/\sigma)$, where the functions $C(x, P_0/\sigma)$ is defined in (\ref{eq:def_cpapcity}). From Result (2), we know that for a given $P_0/\sigma$, $C(x, P_0/\sigma)$  is a continuous and monotonically increasing function of its first argument. Since  $C(x', P_0/\sigma) = a' \geq a = C(x, P_0/\sigma)$, it follows that $x' \geq x$. From Lemma 1, we know that $L_2^\xi(x)$ is a monotonically decreasing function of $x$, and therefore $L_2^\xi(x') \leq L_2^\xi(x)$, and hence $b' = C( L_2^\xi(x'), P_0/\sigma) \leq C(L_2^\xi(x), P_0/\sigma) = b$. Similarly, it can also be shown that, if $b' \leq b$ then it must be true that $a' \geq a$. This completes the proof.              
                  \end{proof}
                  
     \begin{lemma}\label{lem:l2symmetric}
           For a given $0\leq \xi \leq 1$~and~$x \in [0,L_1^{\max}(\xi)]$, the function $L_2^\xi(x)$  is symmetric about $\xi=1/2$, i.e.
             \begin{equation}
             \label{symmetric_L2}
             L_2^{\xi} (x) = L_2^{1-\xi} (x),~~0\leq \xi \leq 1, x \in [0,L_1^{\max}(\xi)].
             \end{equation}
             \end{lemma}
             \begin{proof} See Appendix \ref{ap:l2xisym}.
       \end{proof}\par
       Using Lemma~\ref{lem:l2symmetric} along with the definition of the rate region boundary in (\ref{eq:bd_rate_region}) we get the following result.
       \begin{result}\label{res:bdsym}
       The proposed rate region boundary $R_{\textit{ZF}}^{\textit{Bd}}(\bm H , P_0/\sigma, \xi)$ is symmetric about $\xi=1/2$, i.e. 
       \begin{equation}\label{eq:bdsym}
       R_{\textit{ZF}}^{\textit{Bd}}(\bm H , P_0/\sigma, \xi) = R_{\textit{ZF}}^{\textit{Bd}}(\bm H , P_0/\sigma, (1-\xi)),~ \forall \xi \in [0,1].
       \end{equation}   
       \end{result}      
  The following theorem shows that for $0 \leq \xi \leq 1$, the largest rate region is achieved when $\xi =1/2$.                 
              
              \begin{theorem}\label{the:maxathalf}
              For a fixed $\xi \in [0,1]$, \begin{equation}{\label{eq:maxathalf}}
              R_{\text{ZF}}(\bm H , P_0/\sigma, \xi) \subseteq R_{\text{ZF}}(\bm H , P_0/\sigma, 1/2).\end{equation}
              \end{theorem}
              \begin{proof} See Appendix \ref{ap:bndmaxathalf}.  \end{proof}

                The proposed rate region boundary 
       $R_{\textit{ZF}}^{\textit{Bd}}(\bm H, P_0/\sigma, \xi)$ 
       can be used to compute many practical operating points. Consider a case where
                we are interested in finding the largest
                achievable rate pair $(R_1 , R_2 ) $ such that 
                $R_2 = \alpha R_1 $. 
                This operating point could make sense, if for example the average data
                throughput requested by user 2 is $\alpha$ times that of the throughput requested
                by user 1.   \par       
                Moreover, for a given  $\alpha > 0$ 
                and $P_0/\sigma$, the maximum 
                achievable rate pair of the
                form $(r,\alpha r)$ is given 
   by $(R_{\max}^\alpha(\xi), \alpha R_{\max}^{\alpha}(\xi))$ 
   where $R_{\max}^{\alpha}(\xi)$ is defined as
  \begin{equation}
  \label{eq:r_alpha_max}
  R^{\alpha}_{\max}(\xi) \Define
  \max_{r \big \vert (r, \alpha r) \in 
  R_{\textit{ZF}}(\bm H, P_0/\sigma, \xi)} r.
  \end{equation}
  \begin{theorem}\label{the:uniquebnd}
  $R_{\max}^{\alpha}(\xi)$ is unique and 
  $(R_{\max}^{\alpha}(\xi),\alpha  R_{\max}^{\alpha}(\xi))$ 
  lies on the  boundary 
  $R_{\textit{ZF}}^{\text{Bd}}(\bm H, P_0/\sigma, \xi)$. 
  \end{theorem}          
  \begin{proof} See Appendix \ref{ap:uniqbnd}. \end{proof}              
   \begin{remark} \textit{From the proof in Appendix~\ref{ap:uniqbnd} it is clear that Theorem~\ref{the:uniquebnd} is non-trivial as it depends on the monotonicity and continuity of $L_2^\xi(x)$, which is shown in Lemma~\ref{lem:l2_decrease}. If Lemma~\ref{lem:l2_decrease} were not true, Theorem~\ref{the:uniquebnd} would not hold}. \end{remark} 
  \begin{result}\label{res:rmaxalpha_sym_half}
   Using Theorem~\ref{the:uniquebnd} and (\ref{eq:bdsym}) of Result~\ref{res:bdsym} it follows that for a given $\alpha > 0$,  $R_{\max}^{\alpha}(\xi)$ is symmetric about $\xi=1/2$, i.e.
   \begin{equation}
   R_{\max}^{\alpha}(\xi) =  R_{\max}^{\alpha}(1-\xi),~~\forall~\alpha > 0, \xi \in [0,1]. 
   \end{equation}
    \end{result}           
  \begin{corollary}\label{cor:rmaxalphaliesonbd} From the geometrical interpretation of 
  Theorem~\ref{the:uniquebnd} it follows that 
  $(R_{\max}^{\alpha}(\xi), \alpha R_{\max}^{\alpha}(\xi))$ 
  lies on the intersection of the straight 
  line $R_2 = \alpha R_1$ and the rate region boundary
  $R_{\textit{ZF}}^{\text{Bd}}(\bm H, P_0/\sigma, \xi)$. Further, from the Pareto-optimality of the proposed rate region boundary, it follows that there is only a unique point of intersection between the line $R_2 = \alpha R_1$ and $R_{\textit{ZF}}^{\textit{Bd}}(\bm H, P_0/\sigma, \xi)$.\end{corollary}  
  \subsection{Maximum symmetric rate $R^{sym}(\xi)$}\label{subsec:rsym}
   Note that for the special case of $\alpha=1, R_{\max}^\alpha(\xi)$ is nothing but the maximum achievable symmetric rate which we shall denote by
   \begin{equation}
   R^{sym}(\xi) \Define R^{\alpha = 1}_{max}(\xi).
   \end{equation}
   
   From Theorem~\ref{the:uniquebnd} it is clear that the maximum symmetric rate is nothing but the largest rate $R$ such that the rate pair $(R,R)$ lies on the boundary  $R_{\textit{ZF}}^{\textit{Bd}}(\bm H , P_0/\sigma, \xi)$. From the characterization of the boundary points in (\ref{eq:bd_rate_region}), it follows that there must exist $(x, L_2^\xi(x))$ for some $0 \leq x \leq L_1^{\max}(\xi)$ such that
   \begin{equation}
   R = C(x,P_0/\sigma),~~\text{and}~~
   R = C(L_2^\xi(x), P_0/\sigma)
   \end{equation}
   and therefore 
   \begin{equation}\label{eq:x=l2x}
   x = L_2^\xi(x)
   \end{equation}
   since from Result 2 we know that $C(L,P_0/\sigma)$ is a continuous and monotonic function.
   From \ref{eq:rate_region} it follows that there exists a rectangle $Rect(x, L_2^\xi(x),D(\bm H , \xi)) \subset R_{//}(\bm H )$ corresponding to the rate pair $(R, R)$ where $x$ satisfies $\ref{eq:x=l2x}$.  
   
   Since $x=L_2^\xi(x)$ it follows that this rectangle is infact a square. Further, from the definition of $L_2^\xi(x)$ in (\ref{eq:defl2xi}) it follows that this is the largest sized square whose midpoint is at $D(\bm H,\xi)$ and has side length $x$. \par \textit{Hence, the  maximum achievable symmetric rate corresponds to the largest sized square which is completely inside $R_{//}(\bm H)$ and has its midpoint at $D(\bm H,\xi).$}
   
 \section{A Novel Transceiver Architecture}
  \begin{figure}[!t]
                         
                         \center           
                         {
                         
                         \subfigure[\scriptsize Transmitter (Tx) block diagram.]{\includegraphics[width=9cm,height=6cm]{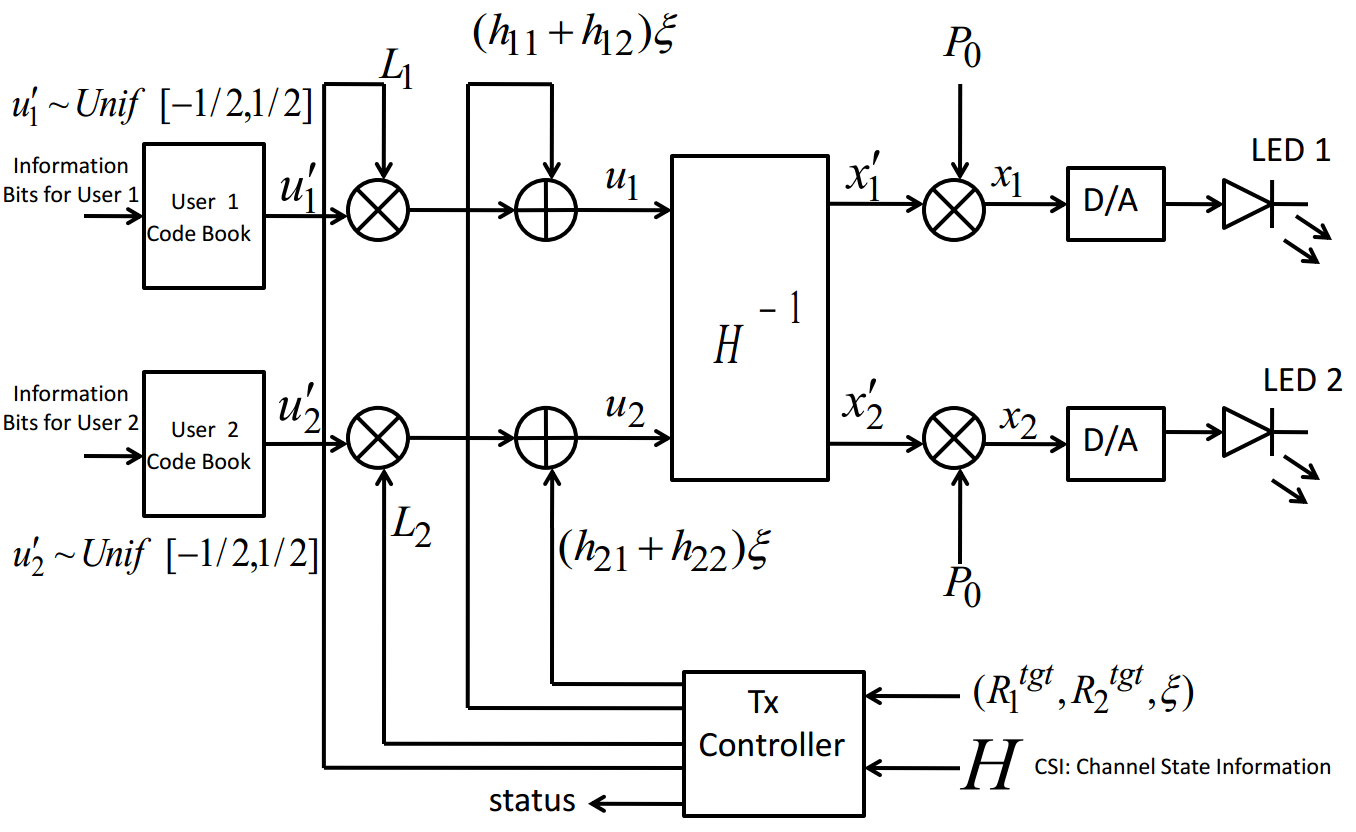}\label{fig:AD1}}
                         
                         \subfigure[\scriptsize Tx Controller block diagram.]{\includegraphics[width=9cm,height=5cm]{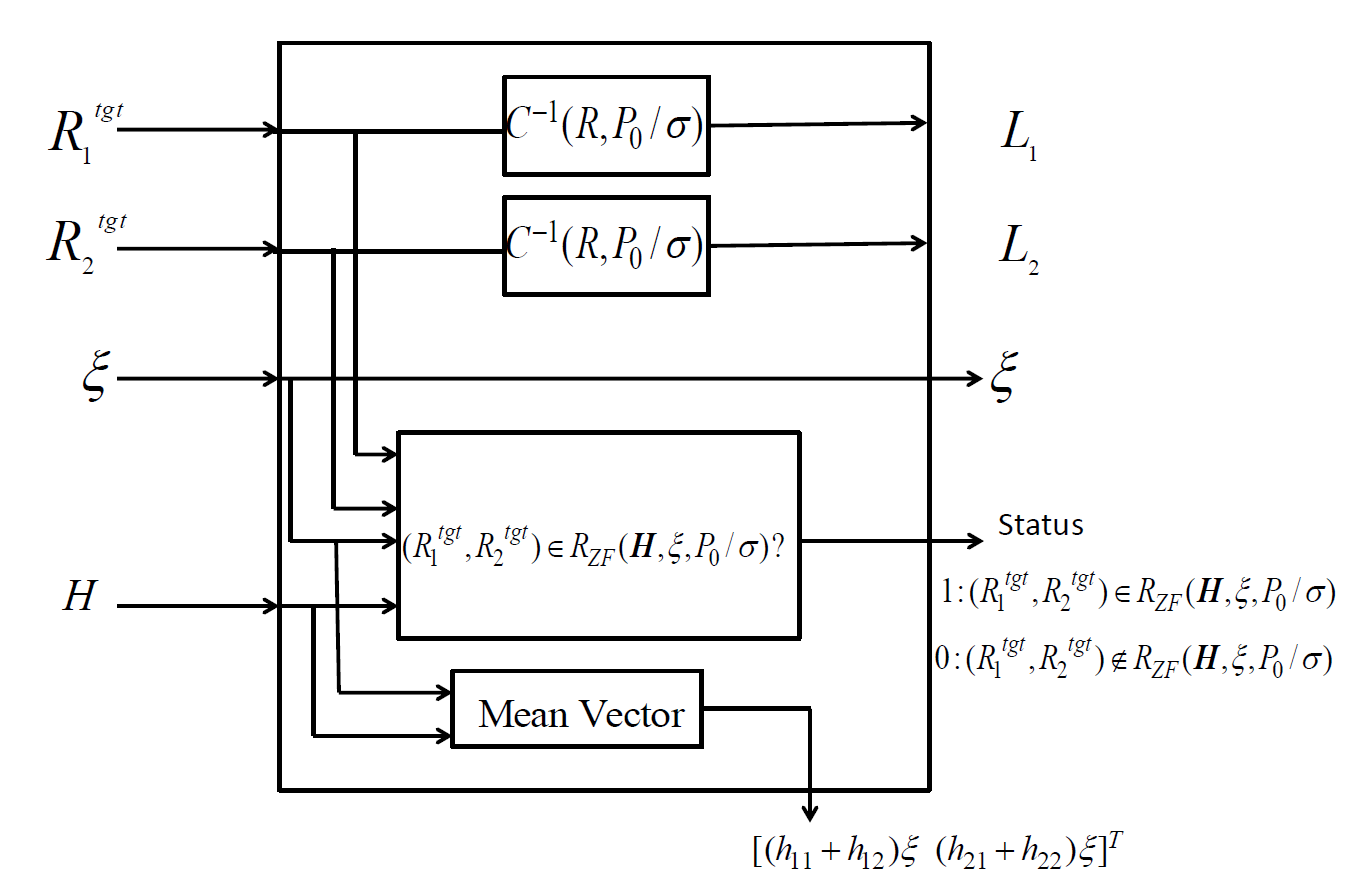}\label{fig:AD2}}
                         
                         \subfigure[\scriptsize Receiver block diagram for $i^{th}$ user,~$i=1,2$. ]{\includegraphics[width=9cm,height=2.5cm]{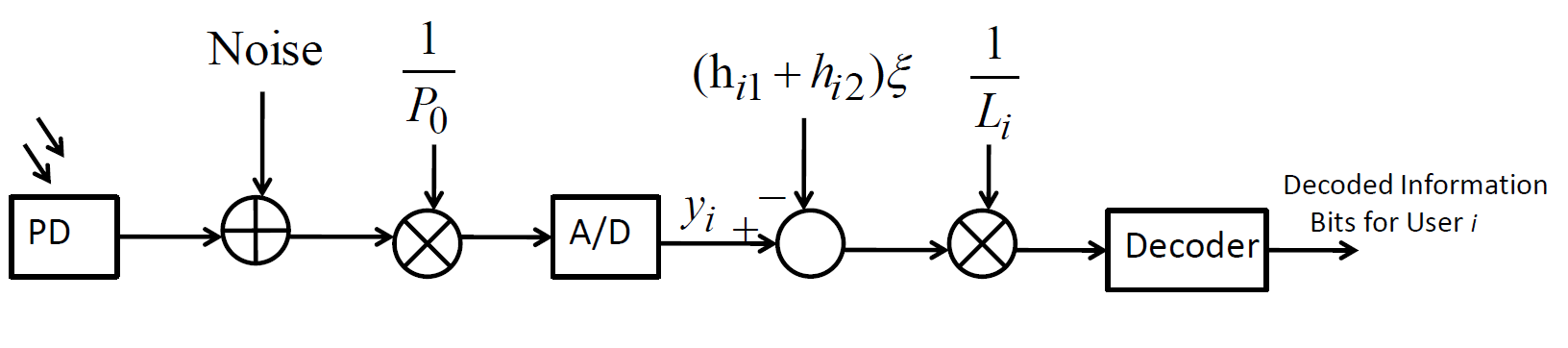}\label{fig:AD3}}

                         \caption{A novel transceiver architecture for the proposed $2\times 2 $ MU-MISO VLC system with dimming Target of $\xi$ and target rate pair $(R_1^{tgt},R_2^{tgt})$.}
                         
                         \label{fig:AD}
                         
                         }
                         
                         \end{figure}
  In this section we propose a novel transceiver architecture for the practical implementation of the proposed $2\times 2$ MU-MISO VLC system  to achieve any rate pair $(R_1, R_2) \in R_{\textit{ZF}}(\bm H , P_0/\sigma, \xi )$ (see Section~\ref{sec:bdzf}), under a per-LED peak power constraint of $P_0$ and a controllable dimming target. \par 
  In Fig.~\ref{fig:AD}, we have shown the block diagram of both the transmitter and the receiver. The block diagram in Fig.~\ref{fig:AD1} depicts the transmitter, the block diagram in Fig.~\ref{fig:AD2} depicts the controller for the transmitter which we call as Tx controller and the block diagram in Fig.~\ref{fig:AD3} depicts the receiver. The working of this transceiver is as follows. \par
 Consider a scenario where the rate requested by User~1 and User~2 are $R_1^{tgt}$ bpcu and $R_2^{tgt}$ bpcu respectively and to satisfy the lighting requirement inside the room the required dimming target is $\xi$. We call this rate pair $(R_1^{tgt}, R_2^{tgt})$, as the target rate pair of the system. The Tx controller first 
 checks if this target rate pair lies in the proposed achievable rate region $R_{\textit{ZF}}(\bm H , P_0/\sigma, \xi )$ (see Section~\ref{sec:bdzf}). If the target rate pair lies inside the proposed achievable rate region, (i.e.,  $(R_1^{tgt}, R_2^{tgt}) \in R_{\textit{ZF}}(\bm H , P_0/\sigma, \xi )$) then the Tx controller flags $1$, otherwise it flags $0$ (see status output of the Tx controller in Fig.~\ref{fig:AD2}). If this flag is 1, then the Tx controller provides $L_1$ and $L_2$, the lengths of the intervals $\mathcal{U}_1$ and $\mathcal{U}_2$. From (\ref{eq:rate_region_equation}) we know that since $(R_1^{tgt}, R_2^{tgt}) \in R_{\textit{ZF}}(\bm H, P_0/\sigma, \xi),$ there must exist some $(L_1,L_2)$ such that $R_1^{tgt} = C(L_1, P_0/\sigma)$ and $R_2^{tgt} = C(L_2, P_0/\sigma)$. From Result~2 we also know that for a given $P_0/\sigma$, $C(x, P_0/\sigma)$ is a monotonic function of $x$, and therefore there exists a corresponding inverse function  $C^{-1}(R, P_0/\sigma)$ such that $C^{-1}(C(L,P_0/\sigma), P_0/\sigma) = L$ and $C(C^{-1}(R,P_0/\sigma), P_0/\sigma) = R$. It then follows $L_i= C^{-1}(R_i^{tgt}, P_0/\sigma),~i =1,2$ see Fig.~\ref{fig:AD3}.
  In the Tx controller we also have a  block which outputs the mean information symbol  vector $\xi(\bm h_1+\bm h_2) = [\xi(h_{11}+h_{12})~~\xi(h_{21}+h_{22})]^T $ (defined in (\ref{eq:mean_info_vec})).
  \par 
  Further, in Fig.~\ref{fig:AD1} the information bits for user 1 and user 2 are coded separately using independent codebooks each having i.i.d. codeword symbols which are unifromly distributed in $[-1/2, 1/2]$.
  The codeword symbols for user~1 and user~2 are denoted by $u_1^{\prime}$
   and $u_2^{\prime}$ respectively (note that $u_1$ and $u_2$ are the information symbols for User ~1 and User~2 respectively). 
  From Section III, we know that the information symbols for the $i^{th}$ user must be uniformly distributed in the interval $\mathcal{U}_i $ i.e.,  $u_i \in \mathcal{U}_i = [ \xi(h_{i1}+h_{i2}) - L_i/2  ,\xi(h_{i1}+h_{i2})+ L_i/2]$ (since the horizontal length of the rectangle corresponding to the rate pair $(R_1^{tgt},R_2^{tgt})$ is $L_1$, the vertical length of this rectangle is $L_2$ and its midpoint is $D(\bm H,\xi)$). Therefore, starting with the codeword symbol $u_i^{\prime}$ we can get the information symbol $u_i$ by
  \begin{equation}
  u_i = L_i u_i^{\prime} + \xi(h_{i1}+h_{i2}),~i=1,2.\end{equation}
  This is also shown in Fig. 3. 
  The information vector $[u_1~u_2]^T$ is then precoded with $H^{-1}$ and scaled by $P_0$ to give the transmit signal vector $[x_1~x_2]^T$.
  It is noted that the proposed transmitter architecture in Fig.~\ref{fig:AD1} allows us to use the \emph{same channel encoder/codebook irrespective of the dimming target $\xi$}. This is because the effect of the dimming control is only in shifting the mean of the information symbols $(u_1,u_2)$ (see the adders in Fig~\ref{fig:AD1}).\footnote{Note that different target rates can be achieved by the same codebook through puncturing of the codewords.}
   \par 
 At the receiver after performing the operations shown in Fig.~\ref{fig:AD3}, we obtain the received vector as given by (\ref{eq:rx_signal}). 
 \section{Numerical Results and Discussions}\label{sec:numres}
 
  \begin{table}[t]
     \centering
      \caption{System Parameters used for Simulation}
      \label{tab:table1}
      \begin{tabular}{|c|c|}   
      \hline
       PD area  & 1 $\text{cm}^2$ \\ \hline
       Receiver Field of Veiw (FOV) & 60 [deg.] \\ \hline
       Refractive index of a lens at the PD & 1.5 \\ \hline
       Semi-angle at half power  & 70 [deg.] \\ \hline
      \end{tabular}   
     \end{table}
 \begin{figure}[t]
                                  \centering
                                  \includegraphics[width=9cm,height=6cm]{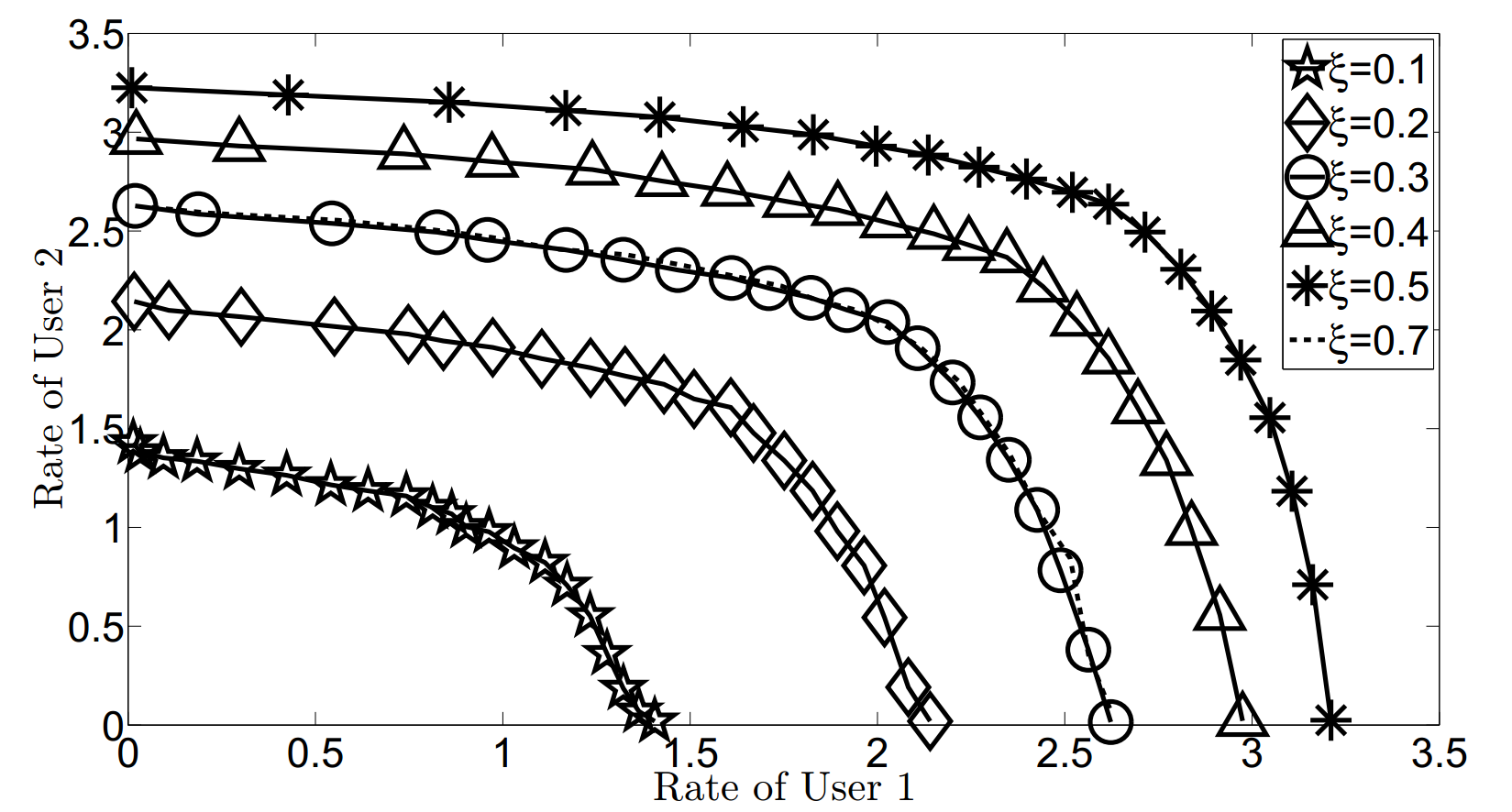}
                                  \caption{Rate region boundary, $R_{\textit{ZF}}^{\textit{Bd}}(\bm H, P_0/\sigma, \xi)$ for different values of dimming target, $\xi$.}
                                  \label{fig:rzfbdwithxi}
                                  \end{figure}      
In this section, we present numerical results in
support of the results reported in previous sections.
For all numerical results we consider an indoor office
room environment where the room is $5$~m $\times~5$~m and its
height is 3 m. The two LEDs are attached to the
ceiling and the two PDs (users) are placed
at a height of 50 cm from the floor of the room. The
two LEDs and the PDs lie in a plane perpendicular
to the floor of the plane. The LEDs are placed
60 cm apart and the ratio $\frac{P_0}{\sigma}$ is
fixed to 70 dB. The channel gains are
modeled for an indoor line of sight (LOS) channel. The other parameters used for simulation are given in Table \ref{tab:table1}. All these parameters and the channel model are taken from prior work \cite{barry_infrared,channel_model_komine,perLEDmu-miso,mu-misosumrate_max}. \par
In Fig.~\ref{fig:rzfbdwithxi}, for a LED separation of 0.6 m and PD (user) separation of 4 m such that the placement of both the LEDs and the PDs is symmetric\footnote{Both the line segment joining the two users and the line segment joining the two LEDs have the same perpendicular bisector.}, we plot the proposed rate region boundary $R_{\textit{ZF}}^{\textit{BD}}(\bm H , P_0/\sigma, \xi)$ (see \ref{eq:bd_rate_region}), for $\xi = 0.1,0.2,0.3,0.4,0.5,0.7.$ \par 
For a given $\xi$, it is observed that the boundary is indeed Pareto-optimal  as is stated in Lemma~\ref{lem:pareto}.  We also observe that as $\xi$ increases from $\xi=0.1$ to $\xi=0.5$, the rate region expands and then it shrinks with further increase in $\xi$ from $\xi =0.5$ onwards to $\xi =1$. We have also observed that rate region boundary is same for both $\xi = 0.3$ and $\xi=1-0.3=0.7$ as is stated in Result~\ref{res:bdsym} (see the dotted line and the solid line marked with circle in Fig.~\ref{fig:rzfbdwithxi}). It is also observed that $\xi=1/2$ gives us the largest rate region as is stated in Theorem~\ref{the:maxathalf}. The expansion/shrinking of the rate region with changing $\xi$ is explained in the following.\par 
For a given $\xi$, the points on the rate region boundary correspond to rectangles in the $u_1-u_2$ plane having their midpoints at $D(\bm H , \xi)$, i.e., on the diagonal $(\bm h_1+\bm h_2)$ and at a distance of $\xi ||\bm h_1+\bm h_2||$ from the origin. As $\xi$ increases, the midpoint of the rectangles move away from the origin and towards the interior of the parallelogram 
$R_{//}(\bm H)$. This allows us to fit bigger rectangles and hence the rate region expands. As $\xi$ is increased beyond $\xi=0.5$  the midpoint of the rectangles moves towards the other end of the diagonal $(\bm h_1+\bm h_2)$ and hence the size of the rectangles reduces thereby shrinking the rate region.\par
\begin{figure}[t]
                                 \centering
                                 \includegraphics[width=9cm,height=6cm]{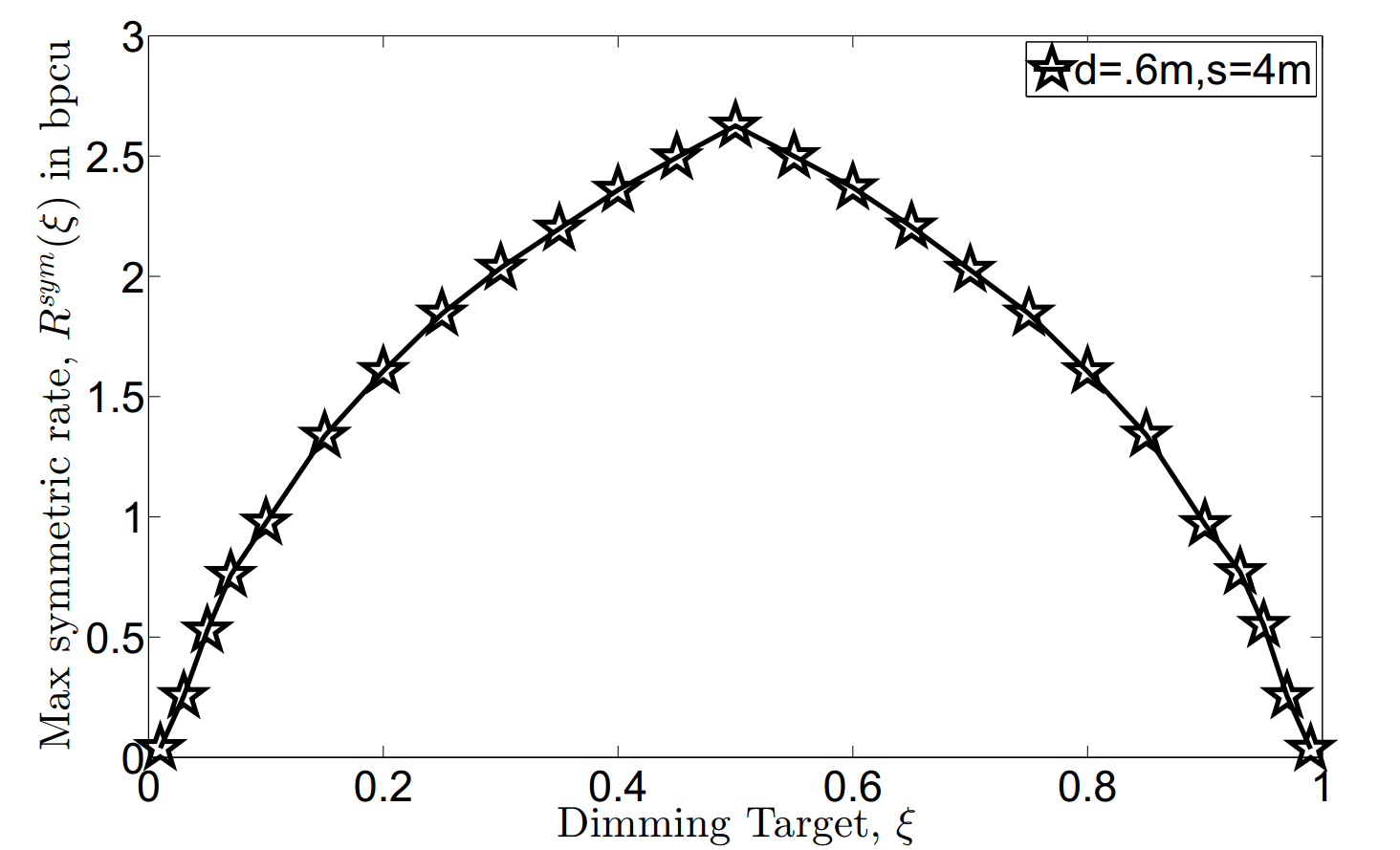}
                                 \caption{Plot between Maximum Symmetric Rate and dimming Target, $\xi$.}
                                 \label{fig:r_sym_max_xi}
                                 \end{figure} 
In Fig.~\ref{fig:r_sym_max_xi}, for a fixed user separation of $s = 4$ m, an LED separation of $d = 60$ cm and symmetric placement of LEDs and PDs,  we plot the maximum achievable symmetric rate $R^{sym}(\xi) \Define R_{max}^{\alpha =1}(\xi)$ as a function of varying $\xi \in [0 , 1]$. We numerically find this operating point by considering all possible points in the $R_1$-$R_2$ plane which lie in the achievable rate region and also lie on the line $R_2 = R_1$. Then among all these possible points we choose the one which has the largest component along the $R_1$ axis. From the figure it is observed that the variation in the maximum symmetric rate with 
 change in the dimming target $\xi$ is small when $\xi$ is around $1/2$, as compared to when $\min(\xi,1 - \xi)$ is small. For example, when $\xi$ is reduced from $\xi=1/2$ to $\xi=0.4$ (i.e., $20\%$ reduction), the corresponding  maximum symmetric rate drops only by $11\%$. However when $\xi$ is reduced by $20\%$ from $\xi=0.07$ to $\xi = 0.056$, the maximum symmetric rate decreases by approximately $25\%$. From this it appears that the maximum symmetric rate is lesser sensitive to variations in the dimming target around $\xi=1/2$ as compared to variations around smaller values of $\xi$. It is also observed that symmetric rate $R^{sym}(\xi)$ is symmetric about $\xi=1/2$ as is stated in Result~\ref{res:rmaxalpha_sym_half} (symmetric rate is nothing but $R_{\max}^{\alpha}(\xi)$ for $\alpha =1$). \par
 \begin{figure}[t] 
                                  \centering
                                  \includegraphics[width=9cm,height=6cm]{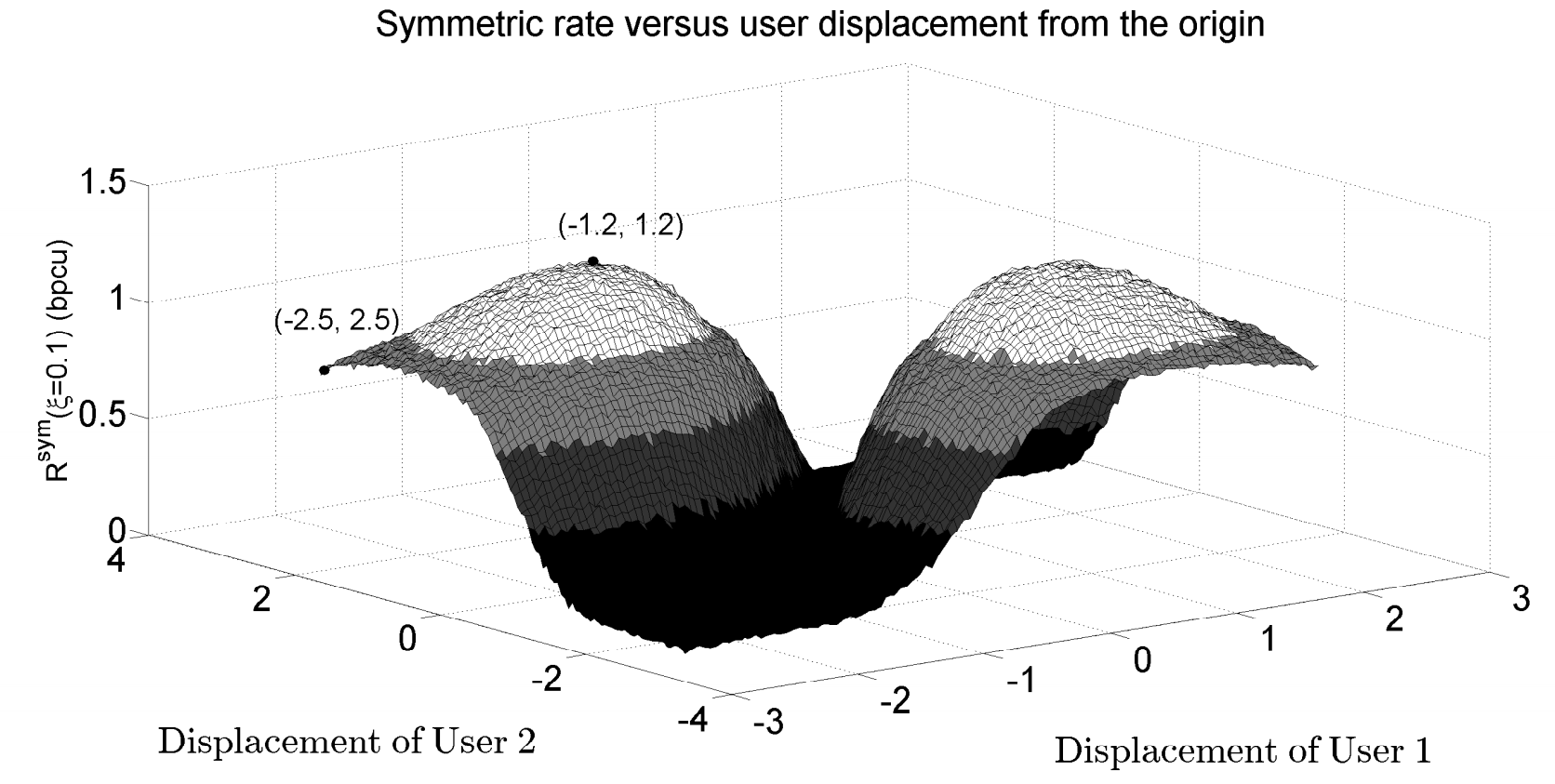}
                                  \caption{Maximum symmetric rate vs displacement of the two users from the origin.}
                                  \label{fig:3d}
                                  \end{figure} 
 We next study the variation in the maximum symmetric rate when the two users (PDs) are  moved along a line parallel to the ceiling (at a height of 50 cm above the floor) while the two LEDs are stationary and fixed to the ceiling with a fixed separation of 60 cm between them and the dimming target is also fixed to $\xi=0.1$. Further, the two LEDs and the two PDs are co-planar. In Fig.~\ref{fig:3d}, we plot the symmetric rate on the vertical axis as a function of the displacement\footnote{Displacement is nothing but the distance of the user from the origin (see Fig.~\ref{fig:setup}.} of the two users from the origin (origin is the point of intersection of the perpendicular bisector of the line joining the LEDs with the line joining the two users, see Fig .~\ref{fig:setup}). In Fig.~\ref{fig:3d} a positive displacement implies that the user PD is located on the right side of the origin and vice versa (see Fig.~\ref{fig:setup}). \par
 \begin{figure}[t] 
                                  \centering
                                  \includegraphics[width=9cm,height=6cm]{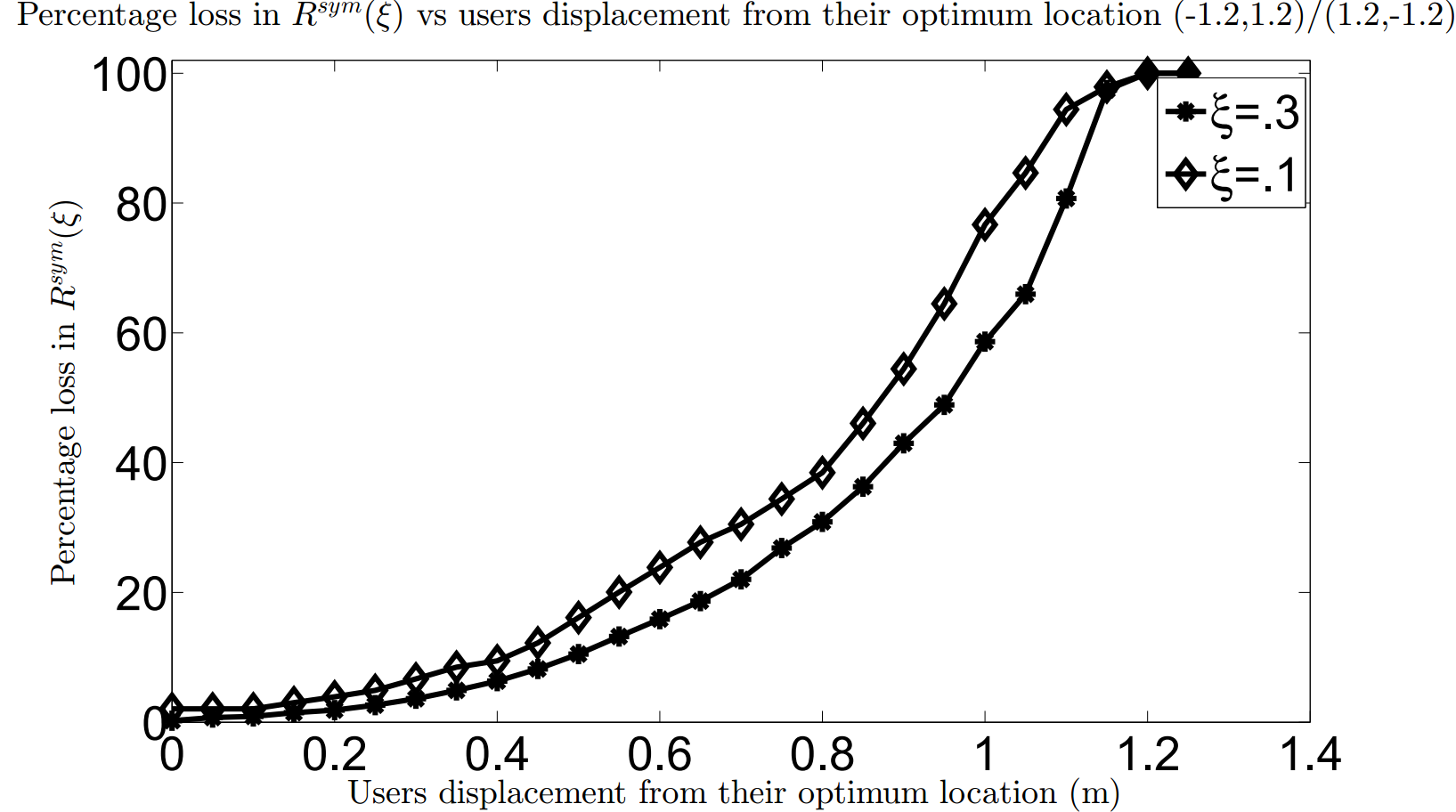}
                                  \caption{Plot between percentage loss in $R^{sym}(\xi)$ and users displacement from their optimum location.}
                                  \label{fig:percent_loss}
                                  \end{figure}       
 It is observed that the maximum symmetric rate is almost zero if the displacement of both the users is same, i.e., the two users are almost co-located. In the figure this is represented by the dark black region. This is expected since in that case the channels to the users is also the same and hence the performance of the ZF precoder degrades. From the figure we observe that starting with both the users at the origin, as user~2 moves towards the right and User~1 moves towards left the maximum symmetric rate increases sharply (in the figure the colour changes from dark black to light black to gray to white as the displacement vector moves from $(0,0)$ to $(-1.2, 1.2)$). This happens because as the users move away from each other, their channels become distinct i.e., the angles between the vectors $\bm h_1$ and $\bm h_2$ increases and hence the area of $R_{//}(\bm H)$ increases. This results in an increase in the largest sized square that can be fit into $R_{//}(\bm H)$ with center at $D(\bm H,\xi)$. This then implies that the maximum symmetric rate would also increase (see the discussion in Section~\ref{subsec:rsym} for the correspondence between the largest square and the maximum symmetric rate). With further increase in the separation between the two users the angular separation between the channel vectors does not increase as sharply as before. At the same time,  due to increased path loss from the LEDs to the users, the area of $R_{//}(\bm H)$ starts decreasing which results in the decrease in the maximum symmetric rate. This can be seen in the figure, as the colour changes back from white to gray, as we move from the displacement vector $(-1.2, 1.2)$ to $(-2.5,2.5)$. This shows that the maximum symmetric rate is dependent on the location of the users and therefore there is an \textit{optimal location}\footnote{By the optimum user location we mean the displacement vector of the users at which we get maximum $R^{sym}(\xi)$.} for both the users which results in the highest symmetric rate. In Fig.~\ref{fig:3d} the optimum location is $(-1.2,1.2)$, or $(1.2,-1.2)$. \par
 
 Next in Fig.~\ref{fig:percent_loss}, for a fixed LED separation of 60 cm we plot the percentage loss in the maximum symmetric rate $R^{sym}(\xi)$ (w.r.t. the symmetric rate at the optimum location) with the users' displacement from their optimum location for two different values of $\xi = 0.1,0.3$.\par
 It is observed that the percentage loss increases with increasing displacement of the PDs from their optimal location. Further, the increase in the percentage loss is small when the displacement is small as compared to when the displacement is large. For example, with $\xi = 0.3$, the percentage loss increases only by $6\%$ as the displacement increases from 0 cm to 40 cm. However with a further increase in displacement from 40 cm to 80 cm, the percentage loss increases sharply from $6\%$ to $30\%$. A similar behavior is also observed with $\xi = 0.1$, though for a given displacement the loss is greater when $\xi = 0.1$ as compared to when $\xi=0.3$.  A practical application of this study could be in defining \textit{coverage zones} for the PDs, i.e., the maximum allowable displacement for a fixed desired upper limit on the percentage loss. For example, in the current setup with $\xi =0.3$, for a $20\%$ upper limit on the percentage loss, the maximum allowable displacement is roughly 70 cm. It therefore appears that indoor VLC systems allow for a lot of flexibility in the movement of the user terminals without significant loss in the information rate.
 
 \section{Conclusion}
 We have proposed an achievable rate region for the  $2\times 2$ MU-MISO broadcast VLC channel under per-LED peak power constraint and dimming control. 
 The boundary of the proposed rate region has been analytically characterized. 
 We propose a novel transceiver architecture to implement such systems. Interestingly, the design of encoder/codebook is independent of the dimming target, which reduces the complexity of the transceiver. Work done in this paper reveals that, in an indoor setting, the two users have enough mobility around their optimal placement without sacrificing their information rates. Our work can also be applied to a 2-D setting, where the users are allowed to move in a plane rather than being restricted to a line.
 
     \begin{appendices}   
     
       \section{Proof of Proposition 1}\label{ap:prop1}
       \begin{proof}
       Under the condition in (\ref{eq:dethlesszero}), to find $L_1^{\text{max}}(\xi)$ we need to consider three scenarios that cover all geometrically possible parallelograms $R_{//}(\bm H)$: (a) ($h_{11} < h_{12}$ and $h_{21} > h_{22})$; (b) $(h_{21} \leq h_{22})$; and (c) $(h_{12} \leq h_{11}$ and $h_{21}>h_{22})$. \par
     
     For a given dimming target, $\xi$, let $L_3$ denote the length of the longest line segment  parallel to the $u_1$-axis lying completely inside $R_{//}(\bm H)$ and whose midpoint coincides with the point $D(\bm H, \xi)$ ($D(\bm H, \xi)$ is defined in (\ref{eq:tip})). For any rectangle $Rect(L_1, L_2, D(\bm H, \xi)) \subset R_{//}(\bm H)$, its side along the $u_1$ axis is a line segment inside $R_{//}(\bm H)$. 
     %and the length of the side has length $L_1$ with midpoint at $D(\bm H, \xi)$. 
     From the definition of $L_3$, it follows that $L_1 \leq L_3$ for any rectangle $Rect(L_1,L_2,D(\bm H,\xi)) \subset R_{//}(\bm H)$. Additionally, the longest line segment of length $L_3$ corresponds to a rectangle $Rect(L_3, L_2 =0,D(\bm H, \xi)) \subset R_{//}(\bm H)$. Hence, it is clear that $L_1^{\max}(\xi) = L_3$, i.e.
     \begin{align}
     L_1^{\text{max}}(\xi)&\triangleq \underset{\underset{Rect(L_1,L_2, D(\bm H, \xi))\subset R_{//}(\bm H)}{L_1 \geq 0, L_2 \geq 0 }}{\max}L_1 \nonumber \\
     &=\max_{\{L_1 > 0  \vert  Rect(L_1, L_2 =0, D(\bm H,\xi)) \subset R_{//}(\bm H)\}}   L_1
     \end{align}
     In the following, we firstly evaluate the expression for $L_1^{\max}(\xi)$ for scenario (a), i.e., when the channel gains satisfy ($h_{11} < h_{12}$ and $h_{21} > h_{22})$. Towards this end, we partition $R_{//}(\bm H)$ into three regions, Region~$i , i=1,2,3,$ as is shown in Fig.~\ref{fig:l1_max}. We now derive an expression for $L_1^{\max}(\xi)$ depending upon the region where $D(\bm H,\xi)$ lies. In Fig.~\ref{fig:l1_max} we denote $D(\bm H, \xi)$ by the point $P$ if $D(\bm H,\xi)$ lies in Region 1, by the point $Q$ if $D(\bm H,\xi)$ lies in Region 2 and by the point $S$ if $D(\bm H,\xi)$ lies in Region 3. Next, we compute $L_1^{\max}(\xi)$ when $D(\bm H, \xi)$ lies in $\text{Region}~1$.\\
      \noindent \textbf{Computation of $L_1^{\max}(\xi)$ when $P=D(\bm H, \xi) \in \text{Region}~1$ : }\vspace{1mm} 
     \hspace{1mm} The point $D(\bm H, \xi)$ belongs to Region 1 if and only if
     \begin{equation}\label{eq:prop1EQ0} 0 \leq OP \leq OT, \end{equation} where the point $T$ denote point of intersection of the diagonal $OB$ and $CC^\prime$. Further, the line $CC^\prime$ is the line parallel to the $u_1$-axis. Next, by looking at the right angle triangle $OT_1T$ in Fig.~\ref{fig:l1_max}, it follows that $OT = TT_1/ \sin\gamma$, where $\gamma$ denotes inclination of the diagonal, $OB$, of the parallelogram $R_{//}(\bm H)$ from the $u_1$-axis. Since, from Fig.~\ref{fig:l1_max}, $TT_1 = h_{22}$, and $\sin\gamma = (h_{21}+h_{22})/OB$, it follows that
     \begin{equation}\label{eq:prop1EQ1} 
     OT = \frac{h_{22}~OB}{h_{21}+h_{22}}
     \end{equation}
     % % % % % % % % % % % % % % % % % % % % % % %
             
     \begin{figure}[!t]
            \centering
            \includegraphics[scale=.3]{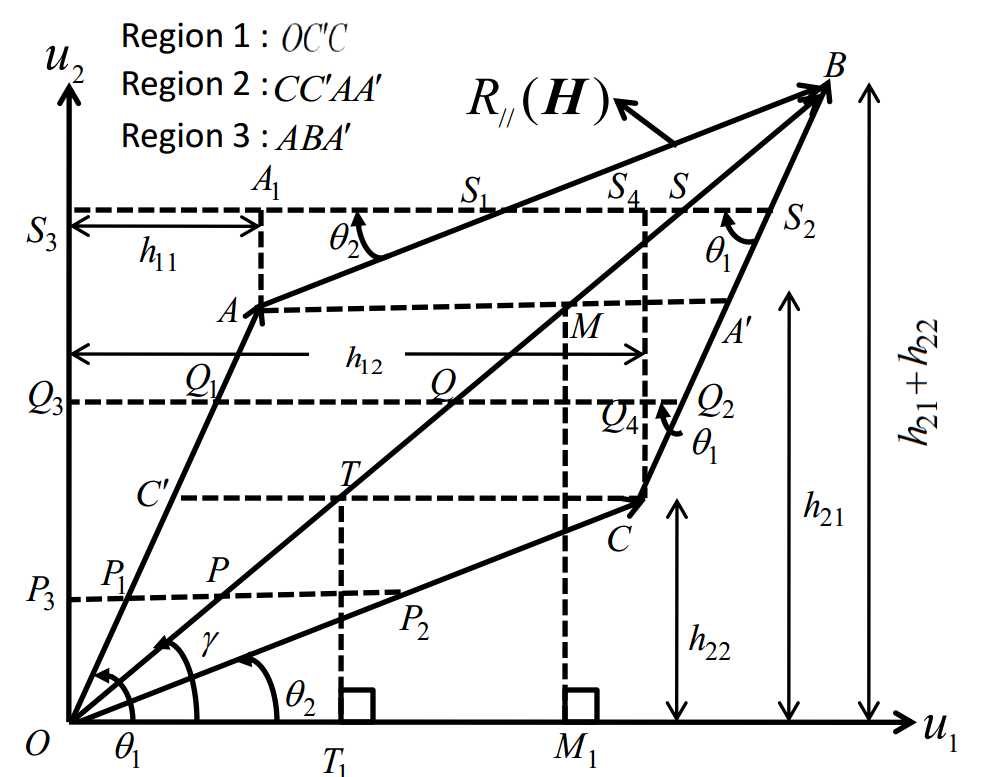}
            \caption{Partition of the parallelogram $OABC$ $\triangleq R_{//}(\bm H)$ into three different regions for the scenario $(h_{11} < h_{12}~\text{and}~ h_{21} > h_{22})$. Note that $AA^{\prime}$ and $CC^{\prime}$ are both parallel to the $u_1$ axis.}
            \label{fig:l1_max} \vspace{-3mm}
            \end{figure}
     
            % % % % % % % % % % % % % % % % % % % % % % % % % %  
     Since the point $P$ is nothing but the point $D(\bm H, \xi)$, from (\ref{eq:tip}), it follows that $OP = \xi~OB$. Using $OP = \xi~OB$ and~(\ref{eq:prop1EQ1}) in (\ref{eq:prop1EQ0}) we have that $D(\bm H, \xi) \in \text{Region}~1$ if and only if 
          \begin{equation}\label{eq:prop1EQ2} 0\leq \xi \leq \frac{h_{22}}{h_{21}+h_{22}}, \end{equation}
          For all such values of the dimming target, $\xi$, satisfying (\ref{eq:prop1EQ2}), it follows that $D(\bm H, \xi) \in \text{Region}~1.$ Next, we evaluate $L_1^{\max}(\xi)$ when $0\leq \xi \leq {h_{22}}/(h_{21}+h_{22})$.\par
          Since $L_1^{\max}(\xi)$ is the length of the line segment parallel to the $u_1$-axis 
          having its midpoint at point $P$ and lying completely inside $R_{//}(\bm H)$.
          it follows that
          \begin{equation}
          \label{eq:propo1EQ3}
          L_1^{\max}(\xi) = 2~\min (PP_1, PP_2),  \end{equation} where both the line segments $PP_1$ and $PP_2$ are parallel to the $u_1$-axis. Further, $P_1$ lies on the line $OA$ whereas $P_2$ lies on the line $OC$ as shown in Fig.~\ref{fig:l1_max}. Next, we evaluate $PP_1$ and $PP_2$. To this end, from Fig.~\ref{fig:l1_max}, we compute the length of the line segment $PP_1$ as follows 
         \begin{align}\label{eq:propo1EQ4}
          PP_1 &= PP_3 -  P_1P_3 \nonumber \\
          &\overset{(a)}{=} \xi (h_{11}+h_{12}) - OP_3/{\tan\theta_1} \nonumber \\
          &\overset{(b)}{=} \xi (h_{11}+h_{12}) - \xi (h_{21}+h_{22})h_{11}/h_{21} \nonumber \\
          &=  \xi (h_{12}h_{21}-h_{11}h_{22})/h_{21} = -\xi~det(\bm H)/h_{21},
          \end{align} 
          where step (a) follows from the fact that, $PP_3$ is equal to the co-ordinate of the point $D(\bm H, \xi)$ along the $u_1$-axis and therefore from (\ref{eq:tip}), we have $PP_3 = \xi (h_{11}+h_{12})$. In step (a) we have also used the fact that since $OP_1P_3$ is a right angle triangle having $\angle OP_3P_1 = \theta_1$. Hence, it follows that $P_1P_3 = OP_3/{\tan\theta_1}$. Step (b) also follows from two facts. Firstly, $OP_3$ is equal to the co-ordinate of the point $D(\bm H, \xi)$ along the $u_2$-axis and therefore from (\ref{eq:tip}), we have that $OP_3 = \xi (h_{21}+h_{22})$ and secondly, from (\ref{eq:tantheta1}), we know that $\tan{\theta_1} = h_{21}/h_{11}$.
          Similarly from Fig.~\ref{fig:l1_max}, we calculate the length of $PP_2$ as follows
          \begin{align} \label{eq:propo1EQ5}
               PP_2 &= P_2P_3 -  PP_3 \nonumber \\
               &\overset{(a)}{=} OP_3/{\tan\theta_2} - \xi (h_{11}+h_{12})  \nonumber \\
               &\overset{(b)}{=} \xi (h_{21}+h_{22})h_{12}/h_{22} - \xi (h_{11}+h_{12})  \nonumber \\
               &=  \xi (h_{12}h_{21}-h_{11}h_{22})/h_{22} = -\xi~det(\bm H)/h_{22},
               \end{align}
             where step (a) follows from the fact that, $PP_3$ is equal to the co-ordinate of the point $D(\bm H, \xi)$ along the $u_1$-axis and therefore from (\ref{eq:tip}), we have that $PP_3 = \xi (h_{11}+h_{12})$. In step (a) we have also used the fact that $OP_3P_2$ is a right angle triangle having $\angle OP_2P_3 = \theta_2$. Hence, it follows that $P_2P_3 = OP_3/{\tan\theta_2}$. Step (b) follows from two facts. Firstly, $OP_3$ is equal to the co-ordinate of the point $D(\bm H, \xi)$ along the $u_2$-axis and therefore from (\ref{eq:tip}), we have that $OP_3 = \xi (h_{21}+h_{22})$ and secondly,  $\tan{\theta_2} = h_{22}/h_{12}$.\\
             \hspace{2mm} Using (\ref{eq:propo1EQ4}) and (\ref{eq:propo1EQ5}) in (\ref{eq:propo1EQ3}) we see that when $0\leq \xi \leq {h_{22}}/({h_{21}+h_{22}}) = {\min(h_{21},h_{22})}/(h_{21}+h_{22})$ (since $h_{21} > h_{22}$ in scenario (a)), we have
     \begin{align}
     L_1^{\max}(\xi) &= -2~\xi~det(\bm H)~\min\left( \frac{1}{h_{21}},\frac{1}{h_{22}}\right) =\frac{ -2~\xi~det(\bm H)}{\max(h_{22},h_{21})}.
     \end{align}
     \vspace{1mm}\\
     \noindent  \textbf{Computation of $L_1^{\max}(\xi)$ when $Q = D(\bm H, \xi) \in \text{Region}~2$ : }\vspace{1mm}
     \hspace{1mm} Point $Q=D(\bm H, \xi)$ lies in Region~2=$CC^{\prime}AA^{\prime}$ if and only if \begin{equation}\label{eq:Dregion2} OT \leq OQ \leq OM, \end{equation} where $M$ is the point of intersection of the line segment $AA^{\prime}$ and the diagonal $OB$ (see Fig.~\ref{fig:l1_max}). In the following we firstly show that $D(\bm H, \xi) \in \text{Region}~2$ if and only if
       \begin{align}\label{eq:xiregion2} 
       \frac{h_{22}}{h_{21}+h_{22}} ~\leq &~\xi \leq \frac{h_{21}}{h_{21}+h_{22}}.
        \end{align}
         Towards this end, we firstly derive an expression for $OM$. From the right angle triangle $OM_1M$ in Fig.~\ref{fig:l1_max}, we know that $OM = MM_1/\sin\gamma$ and since $MM_1 = h_{21}$, $\sin\gamma = (h_{21}+h_{22})/OB$. We have
        \begin{equation}
        \label{eq:om}
        OM = \frac{h_{21}~OB}{h_{21}+h_{22}}.
        \end{equation} Since the point $Q$ is nothing but the point $D(\bm H, \xi)$, from (\ref{eq:tip}), we have  $OQ =  \xi OB$ and from (\ref{eq:prop1EQ1}), it follows that $OT = \frac{h_{22}~OB}{h_{21}+h_{22}}$. In (\ref{eq:Dregion2}), we substitute $OQ$ by $\xi OB$,  $OM$ by the R.H.S in (\ref{eq:om}) and $OT$ by $\frac{h_{22}~OB}{h_{21}+h_{22}}$ to get    
        \begin{align}\label{eq:xiregion2}
        \frac{h_{22}~OB}{(h_{21}+h_{22})} ~\leq ~\xi ~OB \leq \frac{h_{21} ~OB}{(h_{21}+h_{22})} \nonumber \\
        \frac{h_{22} }{(h_{21}+h_{22})} ~\leq ~\xi ~ \leq \frac{h_{21}}{(h_{21}+h_{22})}
        \end{align}
          For all such values of the dimming target, $\xi$, satisfying (\ref{eq:xiregion2}), it follows that $D(\bm H, \xi) \in \text{Region}~2.$ Next, we evaluate $L_1^{\max}(\xi)$ when ${h_{22}}/(h_{21}+h_{22}) ~\leq~ \xi ~\leq~ {h_{21}}/(h_{21}+h_{22})$. For this scenario, from Fig.~\ref{fig:l1_max} we see that \vspace{1mm}
        \small \begin{align}\label{eq:propo1EQ7}
         L_1^{\max}(\xi) = 2~\min(QQ_1, QQ_2), \end{align}\normalsize
          where construction of $QQ_1$ and $QQ_2$ is similar to the construction of $PP_1$ and $PP_2$ (see Fig.~\ref{fig:l1_max}), except the fact that $Q_2$ lies on $CB$ instead of $OC$. Next, we evaluate $QQ_1$ and $QQ_2$. To this end, using the similar steps as for the evaluation of $PP_1$ (see (\ref{eq:propo1EQ4})) we have,
            \small\begin{align}\label{eq:propo1EQ8}
                QQ_1 = QQ_3 -  Q_1Q_3 
                 = \frac{-\xi~det(\bm H)}{h_{21}}.
                \end{align} \normalsize
         However, evaluation of $QQ_2$ is not the same as evaluation of $PP_2$, as the point $Q_2$ lies on the line segment $CB$ whereas the point $P_2$ lies on the line $OC$. Towards this end, using Fig.~\ref{fig:l1_max}, we evaluate $QQ_2$ as follows
        \begin{align}
                    QQ_2 &= Q_3Q_2 - Q_3Q \nonumber \\               
                     &= Q_3Q_4 + Q_4Q_2 - Q_3Q \nonumber \\
                     &= h_{12} + \frac{CQ_4}{\tan\theta_1} - \xi (h_{11} + h_{12}) \nonumber \end{align}
                     \begin{align}\label{eq:propo1EQ9}
                     ~&= h_{12} + \frac{\xi(h_{21} + h_{22}) - h_{22}}{(h_{21}/h_{11})} - \xi (h_{11} + h_{12}) \nonumber \\               
                     ~&= \frac{-~det(\bm H)~(1-\xi)}{h_{21}} 
                    \end{align}\normalsize
       Using (\ref{eq:propo1EQ8}) and (\ref{eq:propo1EQ9}) in (\ref{eq:propo1EQ7}), we see that when ${h_{22}}/(h_{21}+h_{22}) \leq \xi \leq {h_{21}}/(h_{21}+h_{22})$, i.e. $\min(h_{21},h_{22})/(h_{21}+h_{22}) \leq \xi \leq \max(h_{21}, h_{22})/(h_{21}+h_{22})$ (since $h_{21} > h_{22}$ in scenario (a)), we have
         \begin{equation}
         L_1^{\max}(\xi) = \frac{-2~det(\bm H)~\min\left( \xi,(1-\xi) \right)}{ h_{21}} \end{equation} 
         \begin{equation} \overset{(a)}{=} -2~det(\bm H)~\frac{\min(\xi,(1-\xi))}{\max(h_{21}, h_{22})},
         \end{equation}\normalsize             
          where step (a) follows from the fact that $h_{21} = \max(h_{21}, h_{22})$, since for scenario (a), we know that $h_{21}  > h_{22}. $
                  
        \noindent  \textbf{Computation of $L_1^{\max}(\xi)$ when $S = D(\bm H, \xi) \in \text{Region}~3$:}\par
                Point $S=D(\bm H, \xi)$ lies in Region~3=$AA^{\prime}B$ if and only if $OM \leq OS \leq OB$. Using (\ref{eq:om}) $\left(\frac{OM}{OB} = \frac{h_{21}}{h_{21}+h_{22}} \right)$ and the fact that $OS = \xi OB$ (from (\ref{eq:tip})), we have
                \begin{align} \label{eq:xiregion3}
                 {h_{21}}/({h_{21}+h_{22}}) &\leq \xi \leq 1.
                \end{align}
             Next, we evaluate $L_1^{\max}(\xi)$ when ${h_{21}}/({h_{21}+h_{22}}) \leq \xi \leq 1$. From Fig.~\ref{fig:l1_max} it is clear that when $S = D(\bm H, \xi) \in \text{Region}~3$ then $L_1^{\max}(\xi)$ is given by \par
            \small\begin{align}\label{eq:propo1EQ10}
              L_1^{\max}(\xi) &= 2~\min(SS_1, SS_2),
               \end{align}\normalsize
               where $S_1$ and $S_2$ are the intersections of the straight line parallel to the $u_1$ axis passing through $S$, with the line segment $AB$ and $CB$ respectively. Next, we evaluate $SS_1$ and $SS_2$. To this end, using similar steps as for the evaluation of $QQ_2$ (see (\ref{eq:propo1EQ9})) we have,
                \small \begin{align}\label{eq:propo1EQ12}
                                     SS_2 &= S_3S_2 -S_3S \nonumber \\               
                                      &= S_3S_4 + S_4S_2 -S_3S \nonumber \\
                                      &= \frac{-~det(\bm H)~(1-\xi)}{h_{21}} 
                                     \end{align}\normalsize
              However, evaluation of $SS_1$ is not same as evaluation of $QQ_1$, as the point $S_1$ lies on the line segment $AB$ whereas the point $Q_1$ lies on the line $OA$. Towards the evaluation of $SS_1$, using Fig.~\ref{fig:l1_max}, we have
           \small\begin{align}\label{eq:propo1EQ11}
                         SS_1 &= SS_3 - S_1S_3 \nonumber \\               
                         &= \xi(h_{11}+h_{12}) - (S_3A_1 + A_1S_1) \nonumber \\
                         &= \xi (h_{11} + h_{12}) - \left(h_{11} +\frac {AA_1}{\tan\theta_2}\right)\nonumber \\
                         &=\xi (h_{11} + h_{12}) - \left(h_{11} +\frac{\xi(h_{21}+h_{22})- h_{21}}{(h_{22}/h_{12})}\right)\nonumber \\
                          &= \frac{-~det(\bm H)~(1-\xi)}{h_{22}} 
                         \end{align}\normalsize
                        
            Using (\ref{eq:propo1EQ12}) and (\ref{eq:propo1EQ11}) in (\ref{eq:propo1EQ10}), we see that when ${h_{21}}/({h_{21}+h_{22}}) \leq \xi \leq 1 $, i.e. ${\max(h_{21},h_{22})}/({h_{21}+h_{22}}) \leq \xi \leq 1$ (since $h_{21} > h_{22}$ in scenario (a)), we have
              \begin{align}
              L_1^{\max}(\xi) &= -2~det(\bm H) (1-\xi)\min\left(\frac{1}{h_{21}}, \frac{1}{h_{22}}\right) \nonumber \\ &\overset{(a)}{=} \frac{-2~det(\bm H) (1-\xi)}{\max(h_{21}, h_{22})},
              \end{align}            
               where step (a) follows from the fact that $h_{21} = \max(h_{21}, h_{22})$, since for scenario (a), we know that $h_{21}  > h_{22}. $\\   Therefore for the scenario (a), we have the expression of $L_1^{\max}(\xi)$ as follows
                  \begin{align}
                  \label{eq:zero}
                   L_1^{\text{max}}(\xi) &= \left\{\hspace{-3mm}\begin{array}{ll}  
                  \frac {-2 \xi det(\bm H)}{ \max(h_{21}, h_{22})}, & 0 \leq \xi \leq  \eta_1 \triangleq \frac{\min(h_{21}, h_{22})}{(h_{21} + h_{22})}\\ \\
                      \frac {-2det(\bm H) \min(\xi, (1-\xi))} {\max(h_{21}, h_{22})}, & \eta_1 \leq \xi \leq  \eta_2 \triangleq \frac{\max(h_{21}, h_{22})}{(h_{21} + h_{22})}  \\ \\
                       \frac {-2(1-\xi) det(\bm H)}{ \max(h_{21}, h_{22})},&
                                \eta_2 \leq \xi  \leq 1.  \\  
                     \end{array} \right.
                  \end{align} 
                 Since 
                 \begin{equation}\label{eq:one} 
                 \eta_1 \triangleq \frac{\min(h_{21}, h_{22})}{(h_{21} + h_{22})}  \leq 1/2\end{equation} and 
                 \begin{equation}\label{eq:two} 
                 \eta_2 \triangleq \frac{\max(h_{21}, h_{22})}{(h_{21} + h_{22})} \geq 1/2, \end{equation} where the above two inequalities follows from the simple mathematical manipulations, and therefore  the expression in (\ref{eq:zero}) can be  further simplified. To this end, we consider two cases based on the  values of $\xi$.\par
                 \textit{Case (a): $0 \leq \xi \leq 1/2$}\par
                  \noindent
                  For this case we know that $\xi \leq (1-\xi)$ and hence 
                  \begin{equation}
                  \label{eq:three}
                 \min(\xi,(1-\xi)) = \xi 
                  \end{equation}              
                  Since from (\ref{eq:one}) and (\ref{eq:two}), we know that $\eta_1 \leq 1/2$ and $\eta_2 \geq 1/2$ and hence, for $0 \leq \xi \leq 1/2$ from
                  (\ref{eq:zero}) and (\ref{eq:three}) we have             
                 \begin{equation}\label{eq:five} L_1^{\text{max}}(\xi) =  \frac {-2 \xi det(\bm H)}{ \max(h_{21}, h_{22})} \end{equation}\par    
                  \textit{Case (b):  $ 1/2 \leq \xi \leq 1$}\par
                                \noindent For this case we know that $\xi \geq (1-\xi)$ and hence 
                              \begin{equation}
                               \label{eq:four}
                                 \min(\xi,(1-\xi)) = (1-\xi) 
                                 \end{equation}
                Since from (\ref{eq:one}) and (\ref{eq:two}), we know that $\eta_1 \leq 1/2$ and $\eta_2 \geq 1/2$ and hence, for $1/2 \leq \xi \leq 1$ from (\ref{eq:zero}) and (\ref{eq:four}) we have \begin{equation}\label{eq:six} L_1^{\text{max}}(\xi) =  \frac {-2 (1-\xi) det(\bm H)}{ \max(h_{21}, h_{22})}.\end{equation}
                  Therefore from (\ref{eq:five}) and (\ref{eq:six}) the final expression of $L_1^{\max}(\xi)$ for scenario (a) is as follows
                  \begin{align}
                       \label{eq:seven}
                        L_1^{\text{max}}(\xi)&= \left\{\hspace{-3mm}\begin{array}{ll}  
                       \frac {-2 \xi det(\bm H)}{ \max(h_{21}, h_{22})}, & 0 \leq \xi \leq  1/2 \vspace{1mm} \\
                                      \frac {-2(1-\xi) det(\bm H)}{ \max(h_{21}, h_{22})},& 1/2 \leq \xi  \leq 1  
                          \end{array} \right.
                       \end{align}\par

               \noindent Using similar arguments as for scenario (a), we evaluate $L_1^{\max}(\xi)$ for  \textbf{Scenario (b):} $(h_{21} \leq h_{22})$; and for \textbf{Scenario (c):} $(h_{12} \leq h_{11}$ and $h_{21}>h_{22})$ as follows. \par
                To this end, we first partition $R_{//}(\bm H)$ into three regions as shown in 
                Fig.~\ref{fig:proposition1scenariob}. Next, we denote $D(\bm H, \xi)$ by the point $P$ if $D(\bm H,\xi)$ lies in Region 1, by the point $Q$ if $D(\bm H,\xi)$ lies in Region 2 and by the point $S$ if $D(\bm H,\xi)$ lies in Region 3. \par 
              Using (\ref{eq:om}) and (\ref{eq:prop1EQ1}) we can also show that the point $D(\bm H, \xi)$ lies in $\text{Region}~1$ if and only if $0\leq \xi \leq  \frac{\min(h_{21}, h_{22})}{h_{21}+h_{22}}$, $D(\bm H, \xi)$ lies in $\text{Region}~2$ iff $\frac{\min(h_{21}, h_{22})}{h_{21}+h_{22}} \leq  \xi \leq \frac{\max(h_{21}, h_{22})}{h_{21}+h_{22}}$ and it lies in $\text{Region}~3$ iff $\frac{\max(h_{21}, h_{22})}{h_{21}+h_{22}} \leq  \xi \leq 1.$ Next, we evaluate $L_1^{\max}(\xi)$ when $D(\bm H, \xi)$ lies in Region~$i,~i=1,2,3$. \par            
                   Following similar steps as for scenario (a), from Fig.~               \ref{fig:proposition1scenariob} it follows that when $D(\bm H, \xi) \in \text{Region}~1$, i.e., when $0\leq \xi \leq \frac{\min(h_{21}, h_{22})}{(h_{21} + h_{22})}$
                  \begin{equation} 
                  L_1^{\max}(\xi) = \frac {-2 \xi det(\bm H)}{ \max(h_{21}, h_{22})}
                  \end{equation}   
                  Similarly, using Fig.~\ref{fig:proposition1scenariob} it can be shown that when $D(\bm H, \xi) \in \text{Region}~2$, i.e., when $\frac{\min(h_{21}, h_{22})}{(h_{21} + h_{22})} \leq \xi \leq \frac{\max(h_{21}, h_{22})}{(h_{21} + h_{22})}$
                  \begin{equation}
                 L_1^{\max}(\xi) = -2~det(\bm H) \frac{\min(\xi, (1-\xi))}{\max(h_{21}, h_{22})}. \end{equation}
                 Further, using Fig.~\ref{fig:proposition1scenariob} it can be shown that when $D(\bm H, \xi) \in \text{Region}~3$, i.e., when $\frac{\max(h_{21}, h_{22})}{(h_{21} + h_{22})} \leq \xi \leq 1$
                              \begin{equation}
                             L_1^{\max}(\xi) = \frac {-2(1-\xi) det(\bm H)}{ \max(h_{21}, h_{22})}. \end{equation}
                                          
       Following the  steps used to arrive at (\ref{eq:seven}) from (\ref{eq:zero}) in scenario(a), for scenario (b) and (c) also we get the same final expression for $L_1^{\max}(\xi)$ as in (\ref{eq:seven}). This completes the proof. \end{proof}
        \begin{figure}[t]
                                      \centering
                                      \includegraphics[width=9cm,height=4cm]{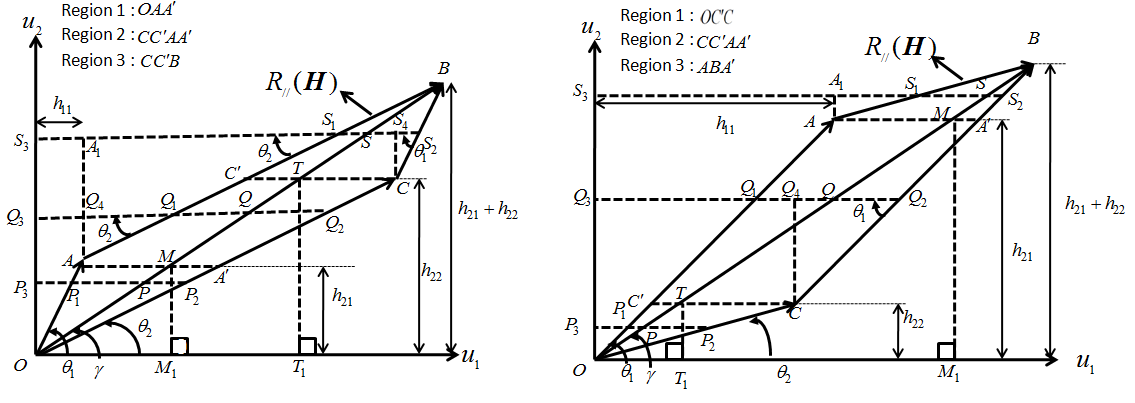}
                                      \caption{Partition of the parallelogram OABC $\triangleq R_{//}(\bm H)$ into three different regions for the scenario (b) $(h_{21} \leq h_{22})$; and scenario
                                            (c) $(h_{12} \leq h_{11}$ and $h_{21}>h_{22})$ (left to right).}
                                      \label{fig:proposition1scenariob}
                                      \end{figure}

       \section{Proof of Proposition 2}\label{ap:prop2}
         \begin{proof}

            Similar to proposition 1, for proving proposition 2, we consider
            three mutually exclusive scenarios. (a) $(h_{11} < h_{12}$ and 
            $h_{21} > h_{22})$; (b) $(h_{21} \leq h_{22})$; and 
            (c) $(h_{12} \leq h_{11}$ and $h_{21}>h_{22})$. Moreover,
            from the definition of LED~1 and LED~2, it follows that, 
            the channel matrix satisfies (\ref{eq:dethlesszero}), i.e. $det(\bm H) < 0.$
            For a fixed $(\bm H, P_0/\sigma, \xi)$, in the following, 
            for a given $L_1 = x, 0\leq x \leq L_1^{\max}(\xi)$, we derive the
            expression for the maximum $L_2$, (i.e., length of the side of the rectangle 
            along the $u_2$-axis (vertical length)), such that there exists a
            rectangle $Rect(L_1 = x, L_2, D(\bm H, \xi) ) \subset R_{//}(\bm H)$, 
            i.e.
            \begin{equation}
            L_2^{\xi}(x) \triangleq 
            \max_{\overset{L_2 \geq 0}{{Rect}(x,L_2,D(\bm H, \xi))\subset 
            \mathcal{R}_{//}(\bm H)}} L_2. \end{equation}                        
             To this end, for a fixed $\xi$,\footnote{Once we fix $\xi$, location of the point $D(\bm H, \xi)$
            gets fixed (see (\ref{eq:tip})).} and given $L_1  = x$, we construct 
            all such possible rectangles $Rect(L_1 = x, L_2, D(\bm H, \xi)) 
            \subset R_{//}(\bm H)$ and among them we choose the rectangle having
            the maximum possible vertical length. \par
            To get this rectangle we first construct a horizontal line segment of 
            the  length $x$ parallel to the $u_1$-axis such that its midpoint coincides with the point $D(\bm H, \xi)$.
             We denote this line segment by 
            $LINE(x,  D(\bm H, \xi))$, i.e.      
            \begin{equation}\label{eq:defline}
            LINE(x, D(\bm H, \xi)) \hspace{-1mm}\Define \hspace{-1mm}\{ v=(v_1,v_2) \in \mathbb{R}^2~|~ v_1 \in S_1, v_2 \in S_2\},
            \end{equation} where $S_1 \Define \{v_1 \in \mathbb{R}~|~|v_1 - \xi(h_{11} + h_{12}) | \leq (x/2)$ and $S_2 \Define \{v_2 \in \mathbb{R} ~|~v_2 = \xi (h_{21} + h_{22})\}.$
              Since $x \leq L_1^{\max}(\xi)$, from the definition of $L_1^{\max}(\xi)$ it follows that $LINE(x,D(\bm H, \xi))$ lies completely inside $R_{//}(\bm H)$.\par
            Given any horizontal line segment of length $x$ parallel to the $u_1$-axis, 
            any rectangle inside $R_{//} (\bm H)$ having this line segment 
            as one of its side  can be constructed in two possible ways, 
             either by extending it vertically downwards or extending it vertically upwards.\footnote{By extending a horizontal 
            line segment vertically downwards/upwards, we mean that we 
            create a rectangle by drawing two vertical lines from the
            end points of this horizontal line segment in the 
            downward/upward direction and then connecting the other two 
            end points of these two vertical lines to form a
            rectangle.} Subsequently, we shall refer to these construction
            methods as ``downward extension'' and ``upward extension''. \par
              
               \begin{figure}[t]
                                             \centering
                                             \includegraphics[scale=.34]{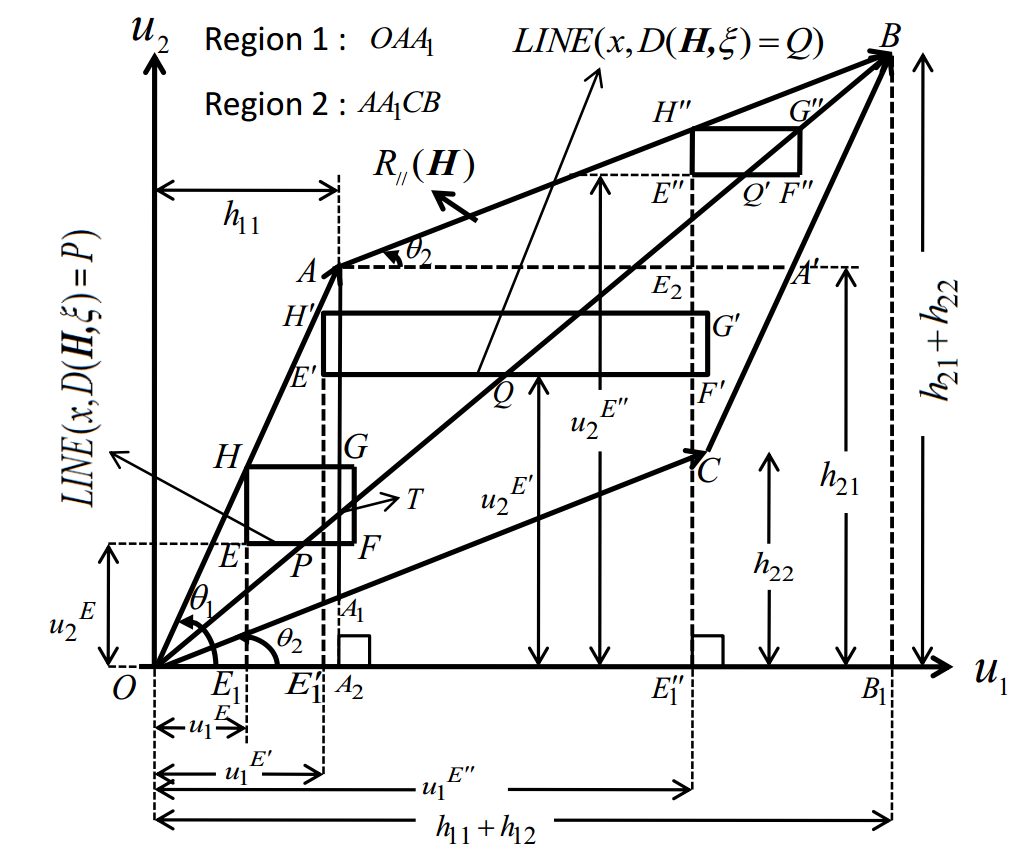}
                                             \caption{ Evaluation of $L_2^{up,\xi}(x)$ for Scenario (a) $(h_{11} < h_{12}~\text{and}~ h_{21} > h_{22})$.}
                                             \label{fig:proposition2L2up}
                                             \end{figure}

  Using this for a given $L_2 > 0$ we can construct a rectangle by extending 
       the line segment 
       $LINE(x,  D(\bm H, \xi))$
       vertically upwards. We denote this rectangle by 
       \begin{equation}
       \label{eq:defrup}
       Rect^{up}( x, L_2, D(\bm H, \xi)) \Define Rect(x, L_2, D(\bm H,\xi) + C_0),
       \end{equation} 
       where $C_0 \Define (0, {L_2}/{2})$. Let $L_2^{up,\xi}(x)$ denote the largest possible vertical length
       of all such rectangles which lie completely inside $R_{//}(\bm H)$ and are constructed by the upward extension of the line segment $LINE (x, D(\bm H, \xi))$
       , i.e.
       \begin{equation}\label{eq:defl2up}
        L_2^{up,\xi}(x)\triangleq
        \max_{\overset{L_2 \geq 0} {Rect^{up}(L_1 = x, L_2, D(\bm H, \xi)) \subset R_{//}(\bm H)}} L_2.
       \end{equation}
       Similarly, we construct any rectangle of vertical length $L_2 \geq 0$ by extending the line segment $LINE (x , D(\bm H, \xi))$ vertically downwards. 
       We denote this rectangle by 
      \begin{equation}\label{eq:defrdown}
      Rect^{down}(x, L_2, D(\bm H, \xi)) \Define Rect(x, L_2, D(\bm H,\xi) + C_1),
             \end{equation} 
       where $C_1 \Define (0,{-L_2}/{2})$. Let $L_2^{down,\xi}(x)$ denote the largest possible vertical length
       of all such rectangles which lie completely inside $R_{//}(\bm H)$ and are  constructed using the ``downward extension''
       method, i.e.
       \begin{equation}\label{eq:defl2down}
        L_2^{down,\xi}(x) \triangleq 
        \max_{\overset{L_2 \geq 0}{Rect^{down}(x, L_2, D(\bm H, \xi)) \subset R_{//}(\bm H)}} L_2.
       \end{equation} 
        $L_2^\xi(x)$  is the maximum possible vertical 
       length of any rectangle lying completely inside $R_{//}(\bm H)$ with horizontal side length 
       equal to $x$ and having its  mid point at 
       $D(\bm H,\xi)$, which is the midpoint of the line segment $LINE(x,  D(\bm H, \xi))$. Equivalently, such a maximal rectangle\footnote{\footnotesize{ By the maximal rectangle we mean, the rectangle with the maximum possible vertical length for a given horizontal length and which lies completely inside the rectangle $R_{//}(\bm H)$}.} must be symmetric about the line segment $LINE(x,  D(\bm H, \xi))$. Since such a maximal rectangle $Rect(x, L_2^\xi(x), D(\bm H, \xi))$ lies inside $R_{//}(\bm H)$, from (\ref{eq:defl2up}) and (\ref{eq:defl2down}) it follows that $$Rect(x, L_2^\xi(x), D(\bm H ,\xi)) \subset S_3 \cup S_4,$$ where $S_3 \Define  Rect^{up}(x, L_2^{up, \xi}(x), D(\bm H, \xi))$ and $S_4 \Define Rect^{down}(x, L_2^{down,\xi}(x), D(\bm H, \xi))$ and $S_3 \cap S_4 = LINE(x, D(\bm H, \xi))$. Further since the maximal rectangle $Rect(x, L_2^\xi(x), D(\bm H , \xi))$ has the maximum possible vertical length and is symmetric about $LINE(x, D(\bm H , \xi))$, it follows that $L_2^\xi(x)/2 =  L_2^{up, \xi}(x)$ if  $L_2^{up, \xi}(x)  \leq L_2^{down, \xi}$,  and $L_2^\xi(x)/2 = L_2^{down, \xi}(x)$ if  $L_2^{down, \xi}(x)  \leq L_2^{up, \xi}(x)$  , i.e.      
       \begin{equation}\label{l2asupdown}
       L_2^{\xi}(x)= 2 
       \min(L_2^{down,\xi}(x),L_2^{up,\xi}(x)).
       \end{equation}
       We next derive expressions for  $L_2^{up,\xi}(x)$ and 
       $L_2^{down,\xi}(x)$ for $0 \leq 
       x \leq L_1^{\max}(\xi)$ for a fixed $0\leq \xi \leq 1$. 
        We firstly consider scenario
       (a) $(h_{11} < h_{12}$~ and~
       $h_{21} > h_{22} )$.
       
       \subsection{Computation of $L_2^{up,\xi}(x)$ for scenario (a)}
             Towards this end, we divide $R_{//} (\bm H)$ 
                   into two regions, Region~$i, i=1,2$, namely Region 1= $OAA_1$ and Region 2 = $AA_1CB$ (see Fig.~\ref{fig:proposition2L2up}). Note that in Fig.~\ref{fig:proposition2L2up}, the straight line $AA_1A_2$ is parallel to the $u_2$-axis and $A_1$ is the point of intersection of this line segment with the side OC of $R_{//}(\bm H)$. Next, we evaluate expressions for $L_2^{up, \xi}(x)$ depending upon the region
             where $D(\bm H, \xi)$ lies. In Fig.~\ref{fig:proposition2L2up}, we denote $D(\bm H, \xi)$ by the point $P$ if $D(\bm H, \xi)$ lies in Region~1 and by the point $Q$/$Q^{\prime}$ if $D(\bm H, \xi)$ lies in Region 2.
              \vspace{1mm}\par
             \noindent\textit{Computation of $L_2^{up, \xi}(x)$ when $D(\bm H, \xi) = P \in$~Region$~1$}:\vspace{1mm} 
             \noindent The point $D(\bm H, \xi) = P 
             \in \text{Region}~1 = OAA_1 $ iff \begin{equation}\label{eq:condr1up}0 \leq OP  \leq OT, \end{equation} where $T$ is the point of intersection of the line segment $AA_2$ with the diagonal $OB$ (see Fig.~\ref{fig:proposition2L2up}). Next, we evaluate expression for $OT$. Towards this end, from the similarity of the triangles $OTA_2$ and $OBB_1$ it follows that $\frac{OT}{OB} = \frac{OA_2}{OB_1}$. Further, from Fig.~\ref{fig:proposition2L2up}, it follows that $OA_2 = h_{11}$ and $OB_1 = h_{11} + h_{12}$ and therefore we have, \begin{equation}\label{eq:ot2}OT = \frac{h_{11}}{h_{11}+h_{12}} OB. \end{equation} Since the point $P$ is nothing but the point $D(\bm H , \xi)$, from (\ref{eq:tip}) we have $OP = \xi
                   OB$. Therefore, using (\ref{eq:ot2}) and $OP = \xi
                               OB$ in (\ref{eq:condr1up}) we have,
              $D(\bm H, \xi) 
             \in \text{Region}~1 $ iff $0\leq \xi \leq \frac{h_{11}}{h_{11}+h_{12}}$. \par
             When 
             $0 \leq \xi \leq \frac{h_{11}}{h_{11}+h_{12}}$, from (\ref{eq:defl2up}) it follows that for evaluating $L_2^{up, \xi}(x)$, we need to construct rectangles $Rect^{up}(x,L_2,D(\bm H, \xi))$ using the ``upward extension'' of the line segment $LINE(x,D(\bm H, \xi) = P) = EF$  as shown in Fig.~\ref{fig:proposition2L2up}. $L_2^{up, \xi}(x)$ is then the largest possible vertical length of all such rectangles which lie inside $R_{//}(\bm H)$. From Fig.~\ref{fig:proposition2L2up}, it is clear that during the upward extension of the line $EF$, with increasing vertical length $L_2$ of the constructed rectangle $Rect^{up}(x, L_2, D(\bm H,\xi))$, the vertically upward line from $E$ will be the first to move out of $R_{//}(\bm H)$ when compared to the vertical line from $F$. Hence it follows that in Fig.~\ref{fig:proposition2L2up}, for $x = EF$ and $0 \leq \xi \leq \frac{h_{11}}{h_{11}+h_{12}}$, we have $L_2^{up,\xi}(x) = EH$. To evaluate $EH$ , we firstly denote the $(u_1,u_2)$ coordinates of the point $E$ by $(u_1^E, u_2^E)$. From the definition of $LINE(x,D(\bm H, \xi)$ and (\ref{eq:tip}) it is clear that 
             \begin{equation} \label{eq:u1e}
                   u_1^E  = \xi (h_{11}+h_{12}) - x/2,~~~
                   u_2^E  = \xi (h_{21}+h_{22}).
                   \end{equation}  
              When 
             $0 \leq \xi \leq \frac{h_{11}}{h_{11}+h_{12}}$ and 
             $ 0 \leq x \leq L_1^{\max}(\xi)$, using Fig.~\ref{fig:proposition2L2up}, $EH$ is computed as follows
            \small{ \begin{align} \label{eq:l2upreg1}
              L_2^{up, \xi}(x) = EH &= E_1H-E_1E \nonumber \\
              &= u_1^E \tan\theta_1 - u_2^E \nonumber \\
              &\overset{(a)}{=} \left(\xi (h_{11}+h_{12}) - \frac{x}{2}\right)
              \frac{ h_{21}}{h_{11}} - \xi (h_{21}+h_{22})  \nonumber \\
              &= \frac{-\xi det(\bm H)- \frac{x}{2} h_{21}}{h_{11}}
             \end{align}}\normalsize
             where $E_1$ is the point of intersection of the extension of the 
             line segment $EH$ and the $u_1$-axis (see Fig.~\ref{fig:proposition2L2up}).
             Step (a) follows from (\ref{eq:u1e}) and (\ref{eq:tantheta1}).\vspace{1mm}\par
             \noindent\textit{Computation of $L_2^{up, \xi}(x)$ when $D(\bm H, \xi) = Q \in$~Region$~2$:}\vspace{1mm}
             \noindent Point $D(\bm H, \xi) = Q$ lies in Region 2 = $AA_1CB$ if and only if
             $OT \leq OQ \leq OB$. Since $Q$ denote the point $D(\bm H, \xi)$, from (\ref{eq:tip}) we have, $OQ = \xi~OB$ and from (\ref{eq:ot2}) we have 
             $OT = \frac{h_{11}OB}{h_{11}+h_{12}}$. Hence, it follows that
             $D(\bm H, \xi)$ lies in Region 2 iff $ \frac{h_{11}}{h_{11}+h_{12}}
             \leq \xi \leq 1.$ We next evaluate
             $L_2^{up, \xi}(x)$ when $\xi$ lies in this interval.\par 
             From (\ref{eq:defl2up}), it follows that for evaluating $L_2^{up, \xi}(x)$, we need to construct rectangles $Rect^{up}(x,L_2,Q=D(\bm H, \xi))$ using the ``upward extension'' of the line segment $LINE(x, Q=D(\bm H,\xi))$ as shown in Fig.~\ref{fig:proposition2L2up}. $L_2^{up,\xi}(x)$ is then the largest possible vertical length of all such rectangles which lie inside $R_{//}(\bm H)$. From Fig.~\ref{fig:proposition2L2up}, it is clear that during the upward extension of the line segment $LINE(x, Q=D(\bm H, \xi))$, the upper left vertex of the constructed rectangle having the largest vertical length will either intersect with the side $OA$ of $R_{//}(\bm H)$ or with the side $AB$ of $R_{//}(\bm H)$ (see rectangles $E^{\prime}F^{\prime}G^{\prime}H^{\prime}$ and $E^{\prime\prime}F^{\prime\prime}G^{\prime\prime}H^{\prime\prime}$ in Fig.~\ref{fig:proposition2L2up}). The upper left vertex intersects with the side $OA$ if and only if  the lower left vertex of the constructed rectangle, (i.e.,~the leftmost point of the line segment $LINE(x,D(\bm H ,\xi))$ (see $E'$ in Fig.~\ref{fig:proposition2L2up}) lies inside Region~1, i.e. 
              \begin{align} \label{eq:conxOA}      
                           u_1^{D(\bm H , \xi)} - x/2 \leq h_{11}, ~\text{i.e.}\nonumber \\
                            2 \xi h_{12} - 2 (1-\xi)h_{11} \leq x ,  
                                           \end{align}
                                          where $u_1$ coordinate of the point $D(\bm H, \xi)$ is denoted by $u_1^{D(\bm H , \xi)}$. From (\ref{eq:tip}), we know that $u_1^{D(\bm H , \xi)} = \xi(h_{11}+h_{12}).$
             On the other hand, the upper left vertex intersects with the side $AB$ of $R_{//}(\bm H)$ if and only if the lower left vertex of the constructed rectangle, i.e., the leftmost point of the line segment $LINE(x,Q^{\prime} = D(\bm H ,\xi))$ (see $E^{\prime\prime}$ in Fig.~\ref{fig:proposition2L2up}) lies inside Region 2, i.e.
             \begin{align}\label{eq:conxAB}
                                      u_1^{D(\bm H , \xi)} - x/2 \geq h_{11},~\text{i.e.} \nonumber \\
                                 x \leq 2 \xi h_{12} - 2 (1-\xi)h_{11}. 
                                   \end{align}                                         
             From the above, we know that when $x$ satisfies (\ref{eq:conxOA}), i.e. $2 \xi h_{12} - 2 (1-\xi)h_{11} \leq x$, the lower left vertex of the  constructed rectangle lies in Region~1 and the upper left vertex lies on he side $OA$. Hence, we have      
             \begin{align}\label{eq:e'h'}
                          L_2^{up,\xi}(x) &= E^{\prime}H^{\prime} =E_1^{\prime}H^{\prime}-E_1^{\prime}E^{\prime}\\             
                            &= \frac{-\xi det(\bm H)- \frac{x}{2} h_{21}}{h_{11}} \nonumber
                                  \end{align} 
                   Similarly, when $x$ satisfies (\ref{eq:conxAB}), i.e. $2 \xi h_{12} - 2 (1-\xi)h_{11} \geq x$, the lower left vertex of the constructed rectangle lies in Region~2 and the upper left vertex lies on the side $AB$. Hence, we have
           \small \begin{align}
                      L_2^{up,\xi}(x) &=  E^{\prime\prime}H^{\prime\prime} \nonumber \\ &=E_1^{\prime\prime}H^{\prime\prime}-E_1^{\prime\prime}E^{\prime\prime} \nonumber \\
                        &= E_1^{\prime\prime}E_2 + E_2H^{\prime\prime} - E_1^{\prime\prime}E^{\prime\prime} \nonumber \\
                       &\overset{(a)}{=} h_{21} + AE_2 \tan\theta_2 - E_1^{\prime\prime}E^{\prime\prime} \nonumber \\
                        &\overset{(b)}{=} h_{21} + (u_1^{E^{\prime\prime}} - h_{11}) \tan\theta_2 - u_2^{E^{\prime\prime}} \nonumber \\
                       &\overset{(c)}{=} h_{21} + \left(\xi (h_{11}+h_{12}) - \frac{x}{2} - h_{11}\right)
                        \frac{ h_{22}}{h_{12}} - \xi (h_{21}+h_{22})  \nonumber \\
                         &= \frac{-(1-\xi) det(\bm H)- \frac{x}{2} h_{22}}{h_{12}},\nonumber
                        \end{align} \normalsize      
                    where $E_2$ is the point of intersection of the line  $E^{\prime\prime} H^{\prime\prime}$ with $AA^{\prime}$ and ($u_1^{E^{\prime\prime}}$, $u_2^{E^{\prime\prime}}$) are the $(u_1,u_2)$ coordinates of the point $E^{\prime\prime}$. Step (a) follows from right angle triangle $AE_2H^{\prime\prime}$. Step (b)
                   follows from the fact that, 
                   $AE_2 = u_1^{E^{\prime\prime}} - h_{11}$ and $E_1E^{\prime\prime} = u_2^{E^{\prime\prime}}$. In step (c) the expression for $u_1^{E^{\prime\prime}}$ and  $u_2^{E^{\prime\prime}}$ follows from 
                    the definition of $LINE(x,D(\bm H, \xi))$ in (\ref{eq:defline}) and (\ref{eq:tip}) and the value of $\tan\theta_2$ follows from
                     (\ref{eq:tantheta1}). Therefore, when $D(\bm H,\xi) \in$ Region 2, (i.e., $\frac{h_{11}}{h_{11}+h_{12}} \leq \xi \leq 1)$ and $(0 \leq x \leq 2\xi h_{12} - 2(1-\xi) h_{11}$, we have
                        \begin{equation}\label{eq:e''h''}
                        L_2^{up,\xi}(x) = \frac{-(1-\xi) det(\bm H)- \frac{x}{2} h_{22}}{h_{12}}
                        \end{equation}         
             Therefore, in scenario (a) $(h_{11} < h_{12}~ \text{and}~ h_{21} > h_{22})$, from (\ref{eq:l2upreg1}), (\ref{eq:e'h'}) and (\ref{eq:e''h''}) we finally have      
            \small {\begin{align}\label{eq:l2upfinala}
                        L^{up,\xi}_2(x)
                              &\hspace{-1mm}=\hspace{-1mm}\left\{\hspace{-2mm}
                                 \begin{array}{ll}
                                 \hspace{-1mm}{-\xi det(\bm H)} - {\frac{x}{2} h_{21}}\over {h_{11}} & \hspace{-2mm}0\leq\xi\leq\mu_1~ \text{and}  ~0 \leq  x \leq L_1^{\max}(\xi) \vspace{1mm}\\
                                 \hspace{-1mm} {-(1-\xi)det(\bm H)}- {\frac{x}{2} h_{22}} \over {h_{12}} &  \mu_1\leq\xi\leq 1~\text{and}~0 \leq x \leq \eta_3(\xi) \vspace{1mm} \\
                                                     \hspace{-1mm} {-\xi det(\bm H)} - {\frac{x}{2} h_{21}}\over {h_{11}} &\hspace{-6mm}\mu_1\leq\xi\leq 1~\text{and}~\eta_3(\xi) \leq x \leq L_1^{\max}(\xi)
                                     \end{array} \right.
                        \end{align}} \normalsize
                        where $\mu_1 \Define \frac{h_{11}}{h_{11}+h_{12}}$ and
                        $\eta_3(\xi) \Define 2\xi h_{12} - 2(1-\xi)h_{11}$. In the next section, we derive expressions for $L_2^{down,\xi}(x)$ for scenario (a).
                         \begin{figure}[t]
                       \centering
                      \includegraphics[width=9cm,height=6cm]{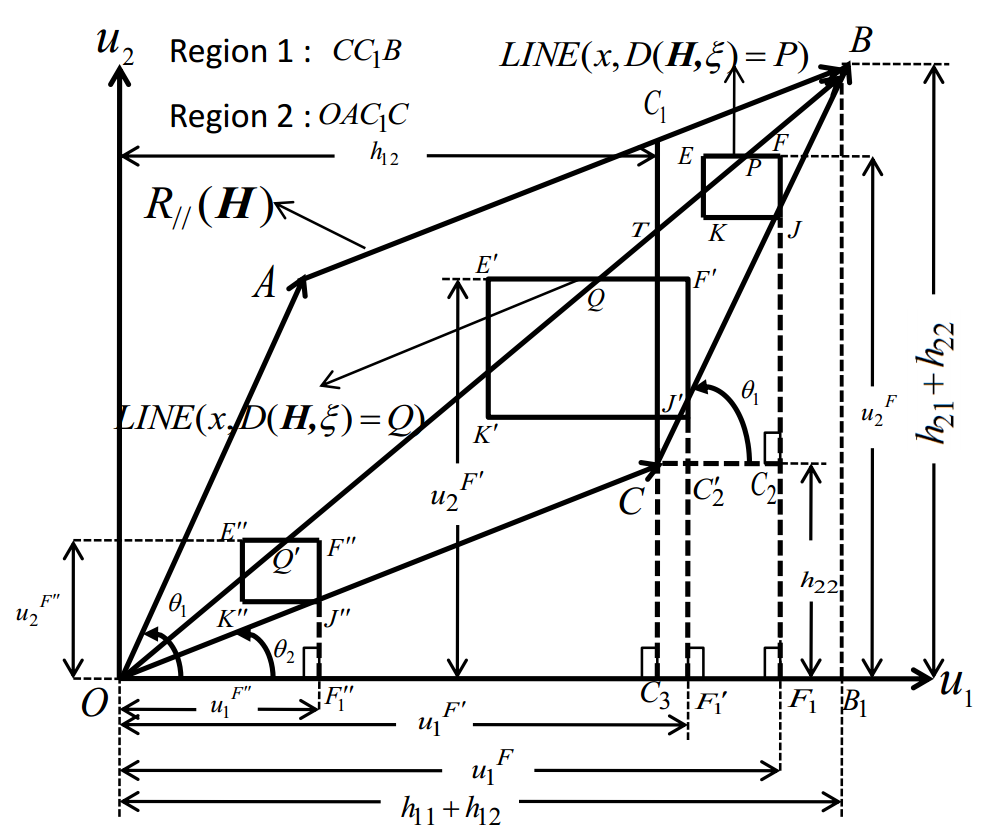}
                       \caption{Evaluation of $L_2^{down,\xi}(x)$ for Scenario (a) $(h_{11} < h_{12}~\text{and}~ h_{21} > h_{22})$.}
                     \label{fig:proposition2L2down}
                     \end{figure}                                     
                                      
          \subsection{Computation of $L_2^{down,\xi}(x)$ for scenario (a)}     
                Evaluation  of $L_2^{down, \xi}(x)$ is similar 
                to that of $L_2^{up, \xi}(x)$. In Fig.~\ref{fig:proposition2L2down}, we partition the 
                parallelogram $R_{//}(\bm H)$ into two regions, namely, Region~1 = 
                $CC_1B$ and Region 2 = $OAC_1C$. Next, we evaluate 
                the expression for $L_2^{down, \xi}(x)$ depending upon the region
                where $D(\bm H, \xi)$ lies. In Fig.~\ref{fig:proposition2L2down}, we denote $D(\bm H, \xi)$ by the point
                $P$ if $D(\bm H, \xi)$ lies in Region~1, by point $Q$/$Q^{\prime}$ if $D(\bm H, \xi)$ lies in Region 2.\vspace{1mm}  \par
                \noindent\textit{Computation of 
                $L_2^{down, \xi}(x)$ when $D(\bm H, \xi) \in$~Region~1:}\vspace{1mm}
                \par \noindent The point $D(\bm H, \xi) = P \in \text{Region}~1 = CC_1B$ iff \begin{equation}\label{eq:condr1down} OT \leq OP  \leq OB, \end{equation} where $T$ is the point of intersection of the straight line $CC_1$ with the diagonal $OB$ (see Fig.~\ref{fig:proposition2L2down}). Note that the straight line $CC_1$ is parallel to the $u_2$ axis (see Fig.~\ref{fig:proposition2L2down}). Next, we derive an expression for $OT$. Towards this end, from the similarity of the triangles $OTC_3$ and $OBB_1$ it follows that \begin{equation}\label{eq:ot2down}OT = \frac{h_{12}}{h_{11}+h_{12}} OB. \end{equation} Since $P=D(\bm H , \xi)$, from (\ref{eq:tip}) we have $OP = \xi
                            OB$. Therefore, using (\ref{eq:ot2down}) and $OP = \xi
                                        OB$ in (\ref{eq:condr1down}) we have, 
                                              $D(\bm H , \xi) \in \text{Region}~1 $ iff $ \frac{h_{12}}{h_{11}+h_{12}} \leq \xi \leq 1$. \par
                                              When 
                                              $ \frac{h_{12}}{h_{11}+h_{12}} \leq \xi \leq 1$, from (\ref{eq:defl2down}), it follows that for evaluating $L_2^{down, \xi}(x)$, we need to construct rectangles $Rect^{down}(x,L_2,D(\bm H, \xi))$ using the ``downward extension'' of the line segment $LINE(x,D(\bm H, \xi) = P) = EF$  as shown in Fig.~\ref{fig:proposition2L2down}. $L_2^{down, \xi}(x)$ is then the largest possible vertical length of all such rectangles which lie inside $R_{//}(\bm H)$. From Fig.~\ref{fig:proposition2L2down}, it is clear that during the downward extension of the line $EF$, with increasing vertical length $L_2$ of the constructed rectangle $Rect^{down}(x, L_2, D(\bm H,\xi))$, the vertically downward line from $F$ will be the first to move out of $R_{//}(\bm H)$ when compared to the vertical line from $E$. Hence, it follows that in Fig.~\ref{fig:proposition2L2down}, for $x = EF$ and $ \frac{h_{12}}{h_{11}+h_{12}} \leq \xi \leq 1$, we have $L_2^{down,\xi}(x) = FJ$. Let us denote the $(u_1,u_2)$ coordinates of the point $F$ by $(u_1^F, u_2^F)$. From the definition of $LINE(x,P= D(\bm H, \xi))$ in (\ref{eq:defline}) and from (\ref{eq:tip}) we have 
                                \begin{equation} \label{eq:u1f}
                u_1^F  = \xi (h_{11}+h_{12}) + \frac{x}{2},~~~
                u_2^F  = \xi (h_{21}+h_{22}).
                \end{equation}
                When 
                      $\frac{h_{12}}{h_{11}+h_{12}} \leq \xi \leq 1$ and 
                      $ 0 \leq x \leq L_1^{\max}(\xi)$, using Fig.~\ref{fig:proposition2L2down}, $FJ$ is computed as follows 
               \small \begin{align}\label{eq:fjoc}
                 L_2^{down, \xi}(x) &= FJ \nonumber \\
                 &= F_1F - F_1J \nonumber \\
                 &= F_1F - F_1C_2-C_2J \nonumber \\
                &\overset{(a)}{=} F_1F - h_{22}- CC_2 \tan\theta_1 \nonumber \\
                 &\overset{(b)}{=} u_2^F - h_{22} - (u_1^F - h_{12}) \tan\theta_1 \nonumber \\
                 &\overset{(c)}{=} \xi(h_{21} + h_{22}) - h_{22} - \hspace{-1mm}\left(\xi(h_{11}+h_{12})+ \frac{x}{2} - h_{12}\right)\hspace{-1mm}\frac{h_{21}}{h_{11}} \nonumber \\
                &= \frac{-(1-\xi) det(\bm H)- \frac{x}{2} h_{21}}{h_{11}},
                      \end{align} \normalsize
                     where $F_1$ is the point of intersection of the extension of the line $FJ$ with the $u_1$-axis (see Fig.~\ref{fig:proposition2L2down}). Step (a) follows from the right angle triangle $CC_2J$. Step (b)
                                 follows from the fact that, 
                                 $CC_2 = u_1^{F} - h_{12}$ and $F_1F = u_2^F$. In step (c) we use the expressions for $u_1^F$ and  $u_2^F$ from 
                                  (\ref{eq:u1f}) and the value of $\tan\theta_1$  from
                                   (\ref{eq:tantheta1}). \par
  \noindent\textit{Computation of 
        $L_2^{down, \xi}(x)$ when $Q = D(\bm H, \xi) \in$~Region$~2$}:
        \par \noindent We know from Fig.~\ref{fig:proposition2L2down}, that $Q = D(\bm H, \xi) 
        \in \text{Region}~2 = OAC_1C$ iff 
        \begin{equation}\label{condxireg2down}0 \leq OQ  \leq OT. \end{equation}
         Since $Q$ is nothing but $D(\bm H ,\xi)$, from (\ref{eq:tip}) we have $OQ = \xi~    OB$. Further, from (\ref{eq:ot2down}) we have $OT = \frac{h_{12}}{h_{11}+h_{12}} OB$. Therefore, it follows that $D(\bm H, \xi) 
        \in \text{Region} ~2 $ iff $  0\leq \xi \leq {h_{12}}/({h_{11}+h_{12}}) $. 
        Next, we evaluate $L_2^{down, \xi}(x)$ when $\xi$ lies in this interval.\par 
        From (\ref{eq:defl2down}), it follows that for evaluating $L_2^{down, \xi}(x)$, we need to construct rectangles $Rect^{down}(x,L_2,D(\bm H, \xi))$ using the ``downward extension'' of the line segment $LINE(x, D(\bm H,\xi))$ as shown in Fig.~\ref{fig:proposition2L2down}. $L_2^{down,\xi}(x)$ is then the largest possible vertical length of all such rectangles which lie inside $R_{//}(\bm H)$. From Fig.~\ref{fig:proposition2L2down}, it is clear that during the downward extension of the line segment $ LINE(x, D(\bm H, \xi))$, the lower right vertex of the constructed rectangle having the largest vertical length will either intersect with the side $OC$ of $R_{//}(\bm H)$ or with the side $CB$ of $R_{//}(\bm H)$ (see rectangles $E^{\prime}F^{\prime}J^{\prime}K^{\prime}$ and $E^{\prime\prime}F^{\prime\prime}J^{\prime\prime}K^{\prime\prime}$ in Fig.~\ref{fig:proposition2L2down}). The lower right vertex intersects with the side $CB$ if and only if  the upper right vertex of the constructed rectangle, (i.e., the rightmost point of the line segment $LINE(x, D(\bm H ,\xi))$ (see $F^{\prime}$ in Fig.~\ref{fig:proposition2L2down}) lies inside Region~1, i.e. 
               \begin{align} \label{eq:conxCB}      
                            u_1^{D(\bm H , \xi)} + x/2 \geq h_{12},~\text{i.e.} \nonumber \\
                            x \geq 2 (1-\xi)h_{12} - 2 \xi h_{11} ,  
               \end{align}
                                           where $u_1^{D(\bm H , \xi)}$ denote the $u_1$ coordinate of the point $D(\bm H, \xi)$. From (\ref{eq:tip}) we know that $u_1^{D(\bm H , \xi)} = \xi(h_{11}+h_{12}).$
              On the other hand, the lower right vertex of the constructed rectangle intersects with the side $OC$ of $R_{//}(\bm H)$ if and only if the upper right vertex of the constructed rectangle (i.e., the rightmost point of the line segment $LINE(x, D(\bm H ,\xi))$) (see $F^{\prime\prime}$ in Fig.~\ref{fig:proposition2L2down}) lies inside Region 2, i.e.
              \begin{align}\label{eq:conxOC}
                                       u_1^{D(\bm H , \xi)} + x/2 \leq h_{12},~~\text{i.e.} \nonumber \\
                                  x \leq   2 (1-\xi)h_{12} - 2 \xi h_{11}. 
                                    \end{align}                                         
              From the above, we know that when $x \leq 2 (1-\xi)h_{12} - 2 \xi h_{11}$, the upper right vertex of the constructed rectangle lies in Region~2 and the lower right vertex lies on the side $OC$. Hence, we have  $L_2^{down,\xi}(x) = F^{\prime\prime}J^{\prime\prime}$. Towards this end, we firstly denote the $(u_1,u_2)$ coordinates of the point $F^{\prime\prime}$ by $(u_1^{F^{\prime\prime}}, u_2^{F^{\prime\prime}})$. From the definition of $LINE(x, D(\bm H, \xi))$ in (\ref{eq:defline}) and from (\ref{eq:tip}) we have 
                                    \begin{equation} \label{eq:u1f''}
                    u_1^{F^{\prime\prime}} = \xi (h_{11}+h_{12}) + \frac{x}{2},~~~
                    u_2^{F^{\prime\prime}}  = \xi (h_{21}+h_{22}).
                    \end{equation}
                    Next, for $0\leq \xi \leq h_{12}/(h_{11}+h_{12})$ and $0 \leq x \leq 2(1-\xi)h_{12} - 2 \xi h_{11}$, we evaluate expression for $L_2^{down,\xi}(x) = F^{\prime\prime}J^{\prime\prime}$ as follows.
           \small   \begin{align}\label{eq:f''j''}
                           L_2^{down,\xi}(x) &= F^{\prime\prime}J^{\prime\prime} \nonumber\\
                           &=F_1^{\prime\prime}F^{\prime\prime}-F_1^{\prime\prime}J^{\prime\prime}\nonumber\\          
                           &\overset{(a)}{=} u_2^{F^{\prime\prime}} - u_1^{F^{\prime\prime}} \tan\theta_2 \nonumber \\
                             &\overset{(b)}{=} \xi(h_{21}+h_{22}) - \left(\xi(h_{11}+h_{12})+\frac{x}{2}\right)\frac{h_{22}}{h_{12}}
                             \nonumber \\
                            &= \frac{-\xi det(\bm H) - \frac{x}{2}h_{22}}{h_{12}},
                                   \end{align}\normalsize 
                   where $F_1^{\prime\prime}$ is the point of intersection of the line $F^{\prime\prime}
                   J^{\prime\prime}$ extended downward with the $u_1$-axis (see Fig.~\ref{fig:proposition2L2down}). Step (a) follows from the two facts. Firstly, from the fact that $F_1^{\prime\prime}F^{\prime\prime}$ is the $u_2$ coordinate  $F^{\prime\prime}$, and secondly from the right angle triangle $OF_1^{\prime\prime}J^{\prime\prime}$, we have $\tan\theta_2 = \frac{F_1^{\prime\prime}J^{\prime\prime}}{u_1^{F^{\prime\prime}}}$, i.e. $F_1^{\prime\prime}J^{\prime\prime}= u_1^{F^{\prime\prime}} \tan\theta_2  $. Step (b) follows from (\ref{eq:u1f''}) and (\ref{eq:tantheta1}).\vspace{1mm}\newline               
                   On the other hand, when $x \geq 2 (1-\xi)h_{12} - 2 \xi h_{11}$, the upper right vertex of the constructed rectangle lies in Region~1 and the lower right vertex lies on the side $CB$. Hence, we have  $L_2^{down,\xi}(x) = F^{\prime}J^{\prime}$. Towards this end, we firstly denote the $(u_1,u_2)$ coordinates of the point $F^{\prime}$ by $(u_1^{F^{\prime}}, u_2^{F^{\prime}})$. From the definition of $LINE(x, D(\bm H, \xi))$ in (\ref{eq:defline}) and from (\ref{eq:tip}) we have 
                                                     \begin{equation} \label{eq:u1f'}
                                     u_1^{F^{\prime}} = \xi (h_{11}+h_{12}) + \frac{x}{2},~~~
                                     u_2^{F^{\prime}}  = \xi (h_{21}+h_{22}).
                                     \end{equation}
                                     The steps involved in the evaluation of $F^{\prime}J^{\prime}$ is exactly the same as for the evaluation of $FJ$ in (\ref{eq:fjoc}). Hence, from Fig.~\ref{fig:proposition2L2down} we have
                             \small  \begin{align}
                                           L_2^{down, \xi}(x) &= F^{\prime}J^{\prime} \nonumber \\
                                           &= F_1^{\prime}F^{\prime} - F_1^{\prime}J^{\prime} \nonumber \end{align}
                                         \small  \begin{align}\label{eq:f'j'}
                                                 ~&= \frac{-(1-\xi) det(\bm H)- \frac{x}{2} h_{21}}{h_{11}},
                                                    \end{align}   \normalsize            
                         
              Therefore, in scenario (a) $(h_{11} < h_{12}~ \text{and}~ h_{21} > h_{22})$, from (\ref{eq:fjoc}), (\ref{eq:f''j''}) and (\ref{eq:f'j'}) we finally have
            \small \begin{align}\label{eq:l2downfinala}
                        \hspace{-1mm} L^{down,\xi}_2\hspace{-.75mm}(x)
                               &\hspace{-1mm}=\hspace{-1.5mm}\left\{\hspace{-3mm}
                                  \begin{array}{ll}
                                  \frac{-(1-\xi) det(\bm H)- \frac{x}{2} h_{21}}{h_{11}}, & \hspace{-3mm} \mu_2\leq \xi\leq 1~ \text{and}  ~0 \leq  x \leq L_1^{\max}(\xi) \vspace{1mm} \\
                                   \frac{-\xi det(\bm H) - \frac{x}{2}h_{22}}{h_{12}}, &  0 \leq\xi\leq \mu_2~\text{and}~0 \leq x \leq \eta_4(\xi) \vspace{1mm} \\
                                                       \frac{-(1-\xi) det(\bm H)- \frac{x}{2} h_{21}}{h_{11}}, & 0 \leq\xi\leq \mu_2~\text{and}~ \\&\eta_4(\xi) \leq x \leq L_1^{\max}(\xi)
                                      \end{array} \right.
                         \end{align}
                         \normalsize
                         where $\mu_2 \Define \frac{h_{12}}{h_{11}+h_{12}}$ and
                         $\eta_4(\xi) \Define 2 (1-\xi)h_{12} - 2 \xi h_{11}$. In the following, we derive expressions for $L_2^{up,\xi}(x)$~and~ $L_2^{down,\xi}(x)$ for scenario (b) $(h_{21} \leq h_{22})$; and 
                               (c) $(h_{12} \leq h_{11}$ and $h_{21}>h_{22})$. \par
                        Similarly, for scenarios (b) and (c) also, by using the \emph{upward} and \emph{downward extension} methods we derive the expressions for $L_2^{up, \xi}(x)$ and $L_2^{down, \xi}(x)$ respectively by constructing rectangles which lie inside $R_{//}(\bm H)$ and have the maximum possible vertical lengths for a given horizontal length. It turns out that the expression for $(L_2^{up, \xi}(x), L_2^{down, \xi}(x))$ is exactly the same as that for scenario (a).

\end{proof}                                      
\section{Proof of Lemma~\ref{lem:l2symmetric}}\label{ap:l2xisym}
\begin{proof}
To prove (\ref{symmetric_L2}) we consider its R.H.S. $L_2^{1-\xi}(x)$. From (\ref{eq:defl2xi}) the R.H.S. is given by
         \begin{align}\label{eq:l2minusxi}
         L_2^{(1-\xi)} (x) &= 2 \min(L_2^{up,(1-\xi)}(x),L_2^{down,(1-\xi)}(x)),  
           \end{align}
           where the expression for $L^{up,(1-\xi)}_2(x)$ is given by, \vspace{1mm} \\
                 \noindent \textit{Case I}: For $0 \leq (1- \xi) \leq \frac{h_{11}}{h_{11}+h_{12}} , \text{i.e.}~ \frac{h_{12}}{h_{11}+h_{12}} \leq \xi \leq 1$ From (\ref{eq:l2up1}) we have
                       \begin{align}
                       L^{up,(1-\xi)}_2(x)
                             &=\left\{\hspace{-2mm}
                                \begin{array}{ll}
                                \frac{{-(1-\xi) det(\bm H)} - {\frac{x}{2} h_{21}}}{h_{11}}, & 0 \leq  x \leq L_1^{\max}(1-\xi)\end{array} \right. \nonumber \\
                           &\overset{(a)}{=}\left\{\hspace{-2mm}
                                                           \begin{array}{ll}\frac{{-(1-\xi) det(\bm H)} - {\frac{x}{2} h_{21}}} {h_{11}}, & 0 \leq  x \leq L_1^{\max}(\xi) 
                                 \end{array}  \right. \nonumber \\
                           &\overset{(b)}{=} L^{down,\xi}_2(x),       
                       \end{align}
                       where step (a) follows from (\ref{eq:symmetric_L1}) and step (b) follows from (\ref{eq:ledown2}). \par
            \noindent\textit{Case II}:  For $\frac{h_{11}}{h_{11}+h_{12}} \leq (1-\xi) \leq 1 ,\text{i.e.}~0 \leq \xi \leq \frac{h_{12}}{h_{11}+h_{12}}$, from (\ref{eq:l2up2}) we have  
            \small \begin{align}
                 L^{up,(1-\xi)}_2(x)
                       &\hspace{-1mm}=\hspace{-1mm}\left\{\hspace{-2mm}
                          \begin{array}{ll}
                          \frac{{-\xi det(\bm H)}- {\frac{x}{2} h_{22}}}{h_{12}},  & 0 \leq x \leq \eta_3 (1-\xi)  \vspace{1mm}\\
                        \hspace{-1mm}\frac  {{-(1-\xi) det(\bm H)} - {\frac{x}{2} h_{21}}}{h_{11}},\hspace{-3mm} & \eta_3 (1-\xi) \leq x \leq L_1^{\max}(1-\xi)                
                                \end{array} \right.\nonumber \\
                      &\overset{(a)}{=}\left\{\hspace{-2mm}
                                               \begin{array}{ll}
                                               \frac {{-\xi det(\bm H)}- {\frac{x}{2} h_{22}}} {h_{12}},
                                               \hspace{-2mm}  & 0 \leq x \leq \eta_4 (\xi)  \vspace{1mm}\\
                                                {-(1-\xi) det(\bm H)} - {\frac{x}{2} h_{21}}\over {h_{11}}, & \eta_4 (\xi) \leq x \leq L_1^{\max}(\xi)                
                                                \end{array} \right.\nonumber \\ 
                      &\overset{(b)}{=}L^{down,\xi}_2(x),                                      
                     \end{align}       \normalsize              
                  where step (a) follows from two facts, firstly that  $\eta_3(1-\xi) \Define 2(1-\xi)h_{12}-2\xi h_{11}=\eta_4(\xi)$ and secondly from (\ref{eq:symmetric_L1}). Step (b) follows from (\ref{eq:l2down1}). Therefore we have,  \begin{equation}\label{eq:l2up=l2down}
                 L^{up,(1-\xi)}_2(x) =  L^{down,\xi}_2(x),~~ \xi \in [0,1], x\in[0,L_1^{\max}(\xi)]  
                 \end{equation} 
                 From (\ref{eq:l2up=l2down}) we also have \begin{equation}\label{eq:l2up=l2down1} L^{up,\xi}_2(x) =  L^{down,(1-\xi)}_2(x),~~ \xi \in [0,1], x\in[0,L_1^{\max}(\xi)]\end{equation} Using (\ref{eq:l2up=l2down}) and (\ref{eq:l2up=l2down1}) in (\ref{eq:l2minusxi}) we finally have 
                 \begin{align}
                 L_2^{(1-\xi)}(x) &= 2 \min(L^{up,(1-\xi)}_2(x),L^{down,(1-\xi)}_2(x)) \nonumber \\
                 &= 2 \min(L^{down,\xi}_2(x),L^{up,\xi}_2(x)) \nonumber \\
                 &\overset{(a)}{=} L_2^{\xi}(x) \nonumber \\ 
                 &= L.H.S.,
                 \end{align} where step (a) follows from (\ref{eq:defl2xi}). This therefore completes the proof. \end{proof} 

\section{Proof of Theorem~\ref{the:maxathalf}}\label{ap:bndmaxathalf}

To Prove this theorem we need the following Lemma.

\begin{lemma} \label{lem:lemmamaxl2}
          For a fixed $x \in [0,L_1^{\max}(1/2)]$, the function $L_2^\xi(x)$ attains its maximum  at $\xi=1/2$, i.e.
          \begin{equation}
                \label{eq:max_L2}
                L_2^\xi(x) \leq L_2^{1/2}(x)~~ \forall~~\xi \in [f(x),1/2]
                \end{equation}  
             where for any $0 \leq x \leq L_1^{\max}(1/2)$, $f(x)$ is the unique\footnote{Uniqueness follows from the fact that $L_1^{\max}(\xi)$ is continuous, increases linearly when $\xi \in [0,1/2]$, has a unique maximum at $\xi = 1/2$, and $L_1^{\max}(\xi)=L_1^{\max}(1-\xi).$} value such that \begin{equation}
             L_1^{\max}(f(x)) = x ~~\text{and}~~ f(x) \leq 1/2.
             \end{equation}
                  \end{lemma}
             \begin{proof} To prove Lemma~\ref{lem:lemmamaxl2} we consider two cases (a) $h_{12} \leq h_{11}$; and (b) $h_{12}\geq
             h_{11}$. From Lemma \ref{lem:l2symmetric}, we know that for a fixed $x\in[0,L_1^{\max}(1/2)]$, $ L_2^\xi(x)$ is symmetric about $\xi=1/2$, hence we consider $\xi$ only in the range [0,1/2]. The proof of Lemma~\ref{lem:lemmamaxl2} is as
             follows \par
             \textit{Case(a)} $h_{12} \leq h_{11}$: \par
             \begin{align}
             h_{12} &\leq h_{11},~~ \text{i.e.}~~~ 1/2 \leq  \frac {h_{11}}{h_{11} + h_{12}}.
             \end{align}
             Since $\frac {h_{11}}{h_{11} + h_{12}}  \geq 1/2$ and $\xi \in [0,1/2]$, we have $\xi \leq \frac {h_{11}}{h_{11} + h_{12}} $. Therefore  from (\ref{eq:l2up1}) we have
           \small  \begin{align}
             L_2^{up,\xi}(x) = 
                                \frac{-\xi det(\bm H)-\frac{x}{2}h_{21}}{h_{11}},~ & 0 \leq x \leq L_1^{\max}(1/2).
             \end{align}\normalsize
             From the above equation it is clear that $L_2^{up,\xi}(x)$ is an increasing function
             of $\xi \in [0,1/2]$ and therefore for any $\xi \in [0,1/2]$ we have
             \begin{align} \label{eq:eq11}
              \xi \leq 1/2 \implies L_2^{up,\xi}(x) \leq L_2^{up,1/2}(x).
             \end{align}
             From (\ref{eq:l2up=l2down1}) we know that $L_2^{up,\xi}(x) = L_2^{down,(1-\xi)}(x)$, and therefore for $\xi = 1/2$
             \begin{equation}\label{eq:l2upxi=l2downxi}L_2^{up,1/2}(x) = L_2^{down,1/2}(x).\end{equation} using
             (\ref{eq:l2upxi=l2downxi}) in (\ref{eq:defl2xi}) we have 
             \begin{align}\label{eq:l2up=l2xiby2}
             L_2^{up,1/2}(x) = L_2^{down,1/2}(x)=L_2^{\xi=1/2}(x)/2
             \end{align} 
             and hence using this along with (\ref{eq:eq11}), for any fixed
             $x = [0,L_1^{\max}(1/2)]$ we have
             \begin{align} \label{eq:eq1}
              L_2^{up,\xi}(x) \leq L_2^{\xi=1/2}(x)/2 ~~, \text{i.e.} \nonumber \\
              2 \min(L_2^{up,\xi}(x),L_2^{down,\xi}(x)) \leq L_2^{\xi=1/2}(x)~~, \text{i.e.} \nonumber \\
              L_2^{\xi}(x) \leq L_2^{\xi=1/2}(x).
             \end{align}
             Therefore for case(a) ($h_{12} \leq h_{11}$) and $\xi \in [0,1/2]$ finally we have
             \begin{equation}
             L_2^{\xi}(x) \leq L_2^{\xi=1/2}(x), ~~ 0 \leq x\leq L_1^{\max}(1/2).
             \end{equation}
             \textit{Case(b)} $h_{11} \leq h_{12}$: \par
             \begin{align}
              h_{11} &\leq h_{12},~~\text{i.e.}~~~  1/2 \leq  \frac {h_{12}}{h_{11} + h_{12}}. 
              \end{align}
             \noindent For this case, in order to prove $L_2^{\xi}(x) \leq L_2^{\xi=1/2}(x)$, we further consider two different
             cases on the basis of the values of $x\in[0,L_1^{\max}(1/2)]$ (b.I): $x\in[0,h_{12} - h_{11}]$; and
             (b.II): $x\in[h_{12} - h_{11},L_1^{\max}(1/2)]$.\par
             case (b.I):~ $x\in[0,h_{12} - h_{11}]$ \par
             \noindent From (\ref{eq:l2down1}) we have
             \begin{align}
             \eta_4(\xi) & = 2 (1-\xi)h_{12} - 2\xi h_{11} \nonumber\\
              & = 2 h_{12} - 2 \xi(h_{11}+h{12}). \end{align}
             From the above equation it is cleat that $\eta_4(\xi)$ is monotonically decreasing with $0\leq \xi \leq 1/2$ and
             therefore we have  
            \begin{align}
            \eta_4(1/2) \leq \eta_4(\xi) \leq \eta_4(0)\nonumber \\
             h_{12} - h_{11} \leq \eta_4(\xi) \leq 2 h_{12}  
            \end{align} and since we know that
             $x\in[0,h_{12} - h_{11}]$, hence for any value of $\xi\in [0,1/2]$, $x$ will always be less than $\eta_4(\xi)$, i.e. $x \leq \eta_4(\xi)$ . Therefore  from
             (\ref{eq:l2down1}) we have
             \begin{align}
             L_2^{down,\xi}(x) = \frac{-\xi det (\bm H)-\frac{x}{2}h_{22}}{h_{12}}.
             \end{align}
             From the above equation it is clear that for a fixed $x\in[0,h_{12} - h_{11}]$,  
              $L_2^{down,\xi}(x)$ is a monotonically increasing function of $\xi \in [0,1/2]$. Therefore
              for any $\xi \in [0,1/2]$ we have
              \begin{align} \label{eq:eq2}
              L_2^{down,\xi}(x) \leq L_2^{down,1/2}(x)= L_2^{\xi=1/2}(x)/2, ~~\text{i.e.} \nonumber \\
              \min(L_2^{up,\xi}(x),L_2^{down,\xi}(x)) \leq L_2^{\xi=1/2}(x)/2, ~~\text{i.e.} \nonumber \\
               2\min(L_2^{up,\xi}(x),L_2^{down,\xi}(x)) \leq L_2^{\xi=1/2}(x), ~~\text{i.e.} \nonumber \\
            L_2^{\xi}(x) \leq L_2^{\xi=1/2}(x).
             \end{align}
             Therefore for case(b.I)  finally we have
             \begin{equation}\label{eq:eqx}
             L_2^{\xi}(x) \leq L_2^{\xi=1/2}(x),~~x \in [0,h_{12}-h_{11}].
             \end{equation}\par
             case (b.II)~$x \in [h_{12} - h_{11} , L_1^{\max}(1/2) ] $: \par
             \noindent From ({\ref{eq:l2up2}}) we have
             \begin{align}
             \eta_3(\xi) & = 2 \xi h_{12} - 2 (1-\xi) h_{11} \nonumber\\
             &= 2 \xi (h_{12}+h_{11}) - 2 h{11} \end{align}
             It is clear from the above equation that $\eta_3(\xi)$ is monotonically increasing with $\xi$ and hence
             for $\xi \in [0,1/2]$ we have
             \begin{align}
              \eta_3(0) &\leq \eta_3(\xi) \leq \eta_3(1/2), ~~\text{i.e.} \nonumber \\
              -2 h_{11} &\leq \eta_3(\xi) \leq h_{12} - h_{11}. \nonumber
             \end{align}
             Since $x \in [h_{12} - h_{11} , L_1^{\max}(1/2) ]$ we have
             \begin{align}
             \eta_3(\xi) \leq h_{12} - h_{11} \leq x,~~\text{i.e.}~~
              \eta_3(\xi) \leq x. \end{align}
              Therefore for case(b.II), from (\ref{eq:l2up1}) and (\ref{eq:l2up2}) we have
              \small\begin{align}
             L_2^{up,\xi}(x) =  \frac{-\xi det(\bm H)-\frac{x}{2}h_{21}}{h_{11}},~&  h_{12} - h_{11}
             \leq x \leq L_1^{\max}(1/2)
             \end{align}\normalsize
             It is clear that $L_2^{up,\xi}(x)$ is  monotonically increasing with $\xi$ and hence using the similar argument as
             for  (\ref{eq:eq1}) we can show that
             \begin{equation} \label{eq:eq3}
             L_2^{\xi}(x) \leq L_2^{\xi=1/2}(x).
             \end{equation}
             Therefore using (\ref{eq:eqx}) and (\ref{eq:eq3})\footnote{In (\ref{eq:eqx}), $ x \in [0,h_{12}-h_{11}]$, whereas in (\ref{eq:eq3}), $x \in [h_{12}-h_{11}, L_1^{\max}(1/2)]$, both of these cases we have the same result and for the union of both of these cases $x \in [0, L_1^{\max}(1/2)]$.} for case (b) ($h_{12} \geq h_{11})$ we have
            \begin{equation} \label{eq:eq4}
             L_2^{\xi}(x) \leq L_2^{\xi=1/2}(x),~0 \leq x \leq L_1^{\max}(1/2).\end{equation}
             Therefore finally from (\ref{eq:eq4}) and (\ref{eq:eq1}) and Lemma~\ref{lem:l2symmetric} we have for any $\xi \in [0,1]$
             $$L_2^{\xi}(x) \leq L_2^{\xi=1/2}(x).$$ This completes the proof of Lemma~\ref{lem:lemmamaxl2}.\end{proof}
            \begin{proof}  Next using this Lemma we prove Theorem~\ref{the:maxathalf}.\par
             Towards this end,
 we consider an arbitrary $\xi \in [0, 1/2]$, for which we show that $R_\textit{ZF}(\bm H , P_0/\sigma, \xi) \subseteq R_{\textit{ZF}}(\bm H , P_0/\sigma,1/2)$. For $\xi \in [1/2,1]$, the proof is similar due to the symmetricity of the $L_2^{\xi}(x)$ and $L_1^{\max}(\xi)$ functions (see Remark~\ref{rem:l1sym} and Lemma~\ref{lem:l2symmetric}). For a given $\xi \in [0,1/2]$ let $(R_1,R_2) \in R_{\textit{ZF}} (\bm H , P_0/\sigma, \xi)$. From (\ref{eq:rate_region}), we know that there exists a $Rect(L_1 \geq 0 , L_2 \geq 0, D(\bm H, \xi) ) \subset R_{//}(\bm H)$ which corresponds to this rate pair $(R_1, R_2)$.  
 Further from proposition~2,
it follows that there exists
\begin{equation}\label{eq:eqb}
L_2^\xi(L_1) \geq L_2.
\end{equation} 
From Lemma~(\ref{lem:lemmamaxl2}) we know that for any a given $\xi \in [ 0, 1/2]$ and 
$0\leq L_1 \leq L_1^{\max}(\xi)$, there exists
\begin{equation}\label{eq:eqc}L_2^{\xi=1/2}(L_1) \geq L_2^{\xi}(L_1)\end{equation}
From (\ref{eq:eqc}) and (\ref{eq:eqb}) we get
\begin{equation}\label{eq:121}
L_2^{\xi=1/2}(L_1) \geq L_2.\end{equation}
We know that for $\xi=1/2$ there exists a rectangle $Rect(L_1, L_2^{\xi=1/2}(L_1),D(\bm H,1/2)) \subset R_{//}(\bm H)$. From (\ref{eq:121}) it therefore follows that there will exists a rectangle $Rect(L_1, L_2,D(\bm H,1/2)) \subset Rect(L_1, L_2^{\xi=1/2}(L_1),D(\bm H,1/2))$ and hence $Rect(L_1, L_2,D(\bm H,1/2)) \subset R_{//}(\bm H)$. The rate pair corresponding to the rectangle $Rect(L_1, L_2,D(\bm H,1/2))$  is $(R_1,R_2)$ and therefore $(R_1,R_2) \in R_{\textit{ZF}}(\bm H , P_0/\sigma, 1/2)$.\end{proof}

\section{Proof of Theorem ~\ref{the:uniquebnd}}\label{ap:uniqbnd}
\begin{proof}
The proof of Theorem~\ref{the:uniquebnd} is as follows. 
Let $(a, \alpha a)$ be any arbitrary rate pair of the 
form $(r, \alpha r)$ lying strictly inside the rate 
region $R_{\textit{ZF}}(\bm H, P_0/\sigma, \xi)$ and  which does not 
lie on the boundary $R_{\textit{ZF}}^{\textit{Bd}}(\bm H, P_0/\sigma, \xi)$ (see the point $P$ in Fig.~\ref{fig:Lemma_2}).

 We then show that there exists the unique rate pair 
 $(a^{\star},\alpha a^{\star})$ which lies on the boundary 
 $R_{\textit{ZF}}^{\textit{Bd}}(\bm H, P_0/\sigma, \xi)$ such 
 that $a^{\star} > a$. This shows that  the rate pair 
 $(R^{\alpha}_{\max}(\xi), \alpha R^{\alpha}_{\max}(\xi)) = 
 (a^\star, \alpha a^\star) $ is the unique rate pair of 
 the form $(r, \alpha r)$ which lies on the boundary.\par

Since, the rate pair $(a,\alpha a) \in R_{\textit{ZF}}(\bm H, 
P_0/\sigma, \xi)$, from (\ref{eq:rate_region}) such a pair $(a,\alpha a)$ will correspond to some rectangle $Rect(y, z > 0, D(\bm H, \xi))$ inside the parallelogram $R_{//}(\bm H)$, where $0 \leq y < L_1^{\max}(\xi)$ such that \begin{equation} \label{eq:lemma3_proof_a_value}
 a = C(y, P_0/\sigma),~~ \text{and} \end{equation} 
 \begin{equation}\label{eq:defalphaa}
   \alpha a = C(z, P_0/\sigma).
   \end{equation}  
 From (\ref{eq:rboud2}), it follows that for the $y$ given in (\ref{eq:lemma3_proof_a_value}), (i.e., for rate of User~1  given in (\ref{eq:lemma3_proof_a_value})) the largest possible rate to the second user is given by 
\begin{equation}\label{eq:defa_1}
 a_1 = C(L_2^\xi(y), P_0/\sigma).\end{equation} 
 From (\ref{eq:bd_rate_region}), it follows that the rate pair $(a,a_1) =  \left(C(y, P_0/\sigma), C(L_2^\xi(y), P_0/\sigma)\right)$ lies on the boundary $R_{\textit{ZF}}^{\textit{Bd}}(\bm H, P_0/\sigma, \xi)$. 
 Since, $a_1$ is the largest possible rate of User~2 when rate of  User~1 is equal to $a$ (see the point $Q$ in Fig.~\ref{fig:Lemma_2}).
 it follows that  \begin{align} \label{eq:lemma3_2} a_1 &> \alpha a \nonumber \\
  C(L_2^\xi(y),P_0/\sigma) &\overset{(a)}{>} C(z,P_0/\sigma), ~\text{i.e.}
   ~~L_2^\xi(y) \overset{(b)}{>} z,
 \end{align}  
 where step (a) follows from (\ref{eq:defa_1}) and (\ref{eq:defalphaa}).
   Step (b) follows from the fact that $C(x, P_0/\sigma)$ is a monotonically increasing function of its first arguments. From the above equation it follows that 
  $z$ lies in the range of the function $L_2^\xi(x)$. It also follows from the continuity and monotonicity of $L_2^\xi(x)$  that there will exist a unique $0 \leq t\leq L_1^{\max}(\xi)$ such that
  \begin{equation}\label{eq:defzt}
  z = L_2^\xi(t)
  \end{equation}
  From the last two equations it follows that
 \begin{equation}
  L_2^\xi(t)
   <  L_2^\xi(y),
  \end{equation} and hence since $L_2^\xi(x)$ is monotonically deceasing
  we have
  \begin{align}\label{eq:123} 
  y &< t \nonumber \\  \
     C(y, P_0/\sigma)&\overset{(a)}{<} C(t, P_0/\sigma),~~\text{i.e.}~~
   a \overset{(b)}{<}  C(t, P_0/\sigma), 
   \end{align}
   where step (a) follows from Result~2 and step (b) follows from (\ref{eq:lemma3_proof_a_value}). We have shown the rate pair $(C(t, P_0/\sigma), \alpha a)$ by the point $R$ in Fig.~\ref{fig:lemma2}.\par 
 \begin{figure}[t]
                                      \centering
                                      \includegraphics[width=6cm,height=6cm]{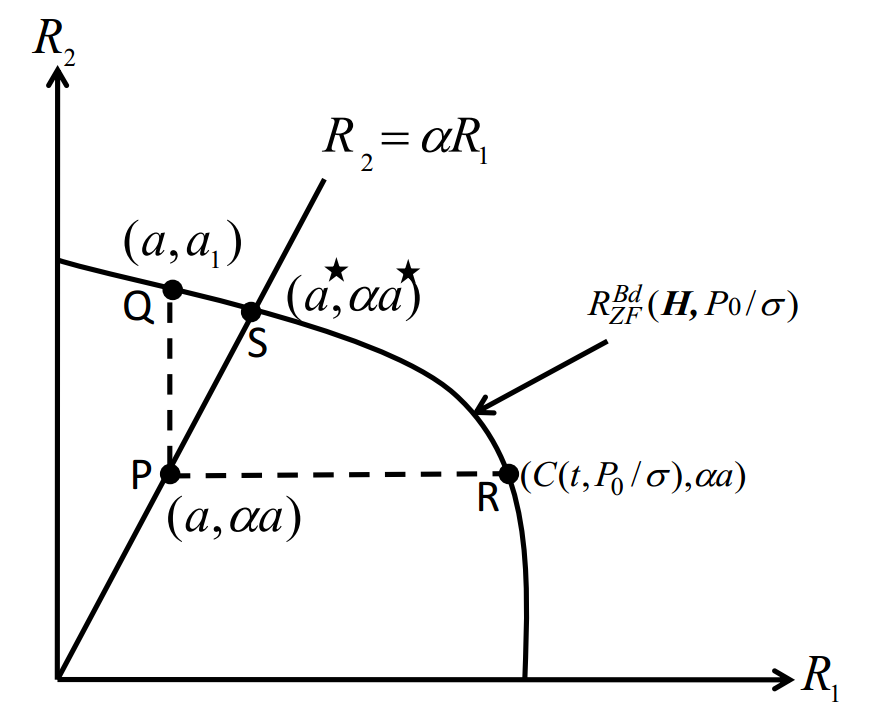}
                                      \caption{A typical proposed rate region boundary. }
                                      \label{fig:lemma2}
                                      \end{figure}   
                                      
 We now define the function \begin{equation} \label{eq:lemma_3_def_f}f(x) \triangleq \alpha 
 C(x, P_0/\sigma) - C(L_2^\xi(x),P_0/\sigma), \end{equation}
 where $x \in [0 , L_1^{\max}(\xi)]$.
 With increasing $x$, $C(x, P_0/\sigma)$ increases  (see Result~2) 
 and
 $C(L_2^\xi(x),P_0/\sigma)$ decreases (as $L_2^\xi(x)$ is a monotonically
 decreasing function of $x$, see Lemma~\ref{lem:l2_decrease}). Hence $f(x)$ is a
 monotonically increasing function of $x$. Further,
 since $C(x, P_0/\sigma)$ and $L_2^\xi(x)$ are continuous functions, 
 it follows
 that $f(x)$ is also continuous.
 It is clear that \small\begin{align} \label{eq:f_y_less_zero}
 f(y) &\overset{(a)}{=} \alpha C(y , P_0/\sigma) - C(L_2^\xi(y) , P_0/\sigma)
       \overset{(b)}{=}  \alpha a  - a_1 \overset{(c)}{<}  0,
 \end{align}\normalsize
 where step (a) follows from (\ref{eq:lemma_3_def_f}), step (b) follows
 from (\ref{eq:lemma3_proof_a_value}) and (\ref{eq:defa_1}), and step (c) follows
 from (\ref{eq:lemma3_2}). Similarly 
 \begin{align}
 \label{eq:f_y_greater_zero}
 f(t) &\overset{(a)}{=} \alpha C(t , P_0/\sigma) - C(L_2^\xi(t) , P_0/\sigma) \nonumber \\
       &\overset{(b)}{=} \alpha C(t , P_0/\sigma) - C(z , P_0/\sigma) \nonumber \\
      &\overset{(c)}{=} \alpha C(t , P_0/\sigma) - \alpha a \nonumber \\ 
       &\overset{}{=} \alpha (C(t , P_0/\sigma) - a) \overset{(d)}{>}  0  
 \end{align}
 where step (a) follows from (\ref{eq:lemma_3_def_f}), step (b) follows from the fact that $L_2^\xi(t) = z$ (see (\ref{eq:defzt})). Step (c) follows from (\ref{eq:defalphaa}). Step (d) follows from (\ref{eq:123}).
  \par Further, 
 since $f(y) < 0,~y\in[0,L_1^{\max}(\xi)]$ and $f(x), x\in[0, L_1^{\max}(\xi)]$ is monotonically 
 increasing in $x$, it follows that $f(x=0) < 0$. 
 Similarly, $f(x = L_1^{\max}(\xi)) > 0$ since $f(t) > 
 0$ and $L_1^{\max}(\xi) \geq t $. Since, $f(x)$ is a 
 monotonically increasing and continuous in $[0, L_1^{\max}(\xi)]$, and $f(0) < 0$, 
 $f(t) > 0$, it follows that there exists a unique 
 $x^{\star} 
 \in [0, L_1^{\max}(\xi)]$ such that $f(x^{\star}) = 0$ \cite{WalterRudin}. 
 The uniqueness follows from the monotonicity of $f(x)$. That is, 
 from (\ref{eq:lemma_3_def_f})
 we have \begin{equation} \label{eq:def_c_x_star}
 \alpha C(x^{\star}, P_0/\sigma) = C(L_2^\xi(x^{\star}),P_0/\sigma). \end{equation} Let $a^{\star} \Define C(x^{\star}, P_0/\sigma)$ and therefore from (\ref{eq:def_c_x_star}) it follows that $\alpha a^\star = 
 C(L_2^\xi(x^{\star}),P_0/\sigma)).$ From (\ref{eq:bd_rate_region}), 
 it is clear that
 the rate pair $(a^{\star}, \alpha a^{\star})
 \in R_{\textit{ZF}}^{\textit{Bd}}(\bm H, P_0/\sigma,\xi)$. 
 Uniqueness of such a rate pair follows from the uniqueness 
 of $x^{\star}$. 
 Further since $f(y) <  0 = f( x^{\star})$ (see Eq.~\ref{eq:f_y_less_zero}) and $f(x)$ is monotonically increasing, it follows that 
 \begin{equation}\label{eq:xstargy} x^{\star} > y.\end{equation}
 From Result (2) we know that $C(x,P_0/\sigma)$ is monotonically increasing in $x$
 and therefore form (\ref{eq:xstargy}) $a^{\star} = C(x^{\star}, P_0/\sigma) > C(y, P_0/\sigma) = a.$  Therefore we have shown that the unique rate pair  $(a^{\star}, \alpha a^{\star})$  lies
 on the boundary $R_{\textit{ZF}}^{\textit{Bd}}(\bm H,P_0/\sigma,\xi)$ 
 and $a^{\star} > a$ for any arbitrary choice of $a$, where $(a,\alpha a)$ lies strictly inside $R_{\textit{ZF}}(\bm H , P_0/\sigma, \xi)$. As shown in Fig.~\ref{fig:Lemma_2}, the point $(a^{\star}, \alpha a^{\star} )$ lies on the line $R_2 = \alpha R_1$ and also on the boundary $R_{\textit{ZF}}^{\textit{Bd}}(\bm H, P_0/\sigma, \xi)$.
 This therefore completes the proof. \end{proof}                                      

\end{appendices}
\bibliographystyle{ieeetr}
%\addbibresource{call_02_Dec.bib}
\bibliography{IEEEabrv,./call_02_Dec}
\end{document}